\documentclass[11pt]{article}
\usepackage{amsfonts,amsthm,amssymb,amsmath,epsfig,color,psfrag,latexsym,graphicx,cite,xspace}
\usepackage[left=1in,right=1in,top=1in,bottom=1in]{geometry}
\usepackage{hyperref}

\newcommand{\floor}[1]{\lfloor #1 \rfloor}
\newcommand{\schur}{\mathrm{s}}
\newcommand{\Tr}{\mathrm{Tr}}

\newcommand{\vecx}{{\mathbf{x}}}
\newcommand{\veca}{{\mathbf{a}}}
\newcommand{\vece}{{\mathbf{e}}}

\newcommand{\perm}{{\mbox {perm}}}
\newcommand{\pf}{\mathrm{pf}}
\newcommand{\sgn}{\mathrm{sgn}}
\newcommand{\sg}{\mathrm{S}}
\newcommand{\coeff}{\mathrm{coeff}}
\newcommand{\npt}{\mathrm{NPT}}

\newcommand{\chartc}{\mathrm{char}}
\newcommand{\poly}{\mathsf{poly}}
\newcommand{\SL}{\mathrm{SL}}
\newcommand{\GL}{\mathrm{GL}}

\newcommand{\rank}{\mathrm{rank}}

\theoremstyle{plain}
\newtheorem{theorem}{Theorem}
\newtheorem{prop}{Proposition}[section]
\newtheorem{cor}[prop]{Corollary}
\newtheorem{lemma}[prop]{Lemma}

\theoremstyle{definition}
\newtheorem{conj}[prop]{Conjecture}
\newtheorem{question}[prop]{Question}
\newtheorem{ex}[prop]{Example}
\newtheorem*{remark}{Remark}

\newcommand{\bitlength}[1]{\langle #1 \rangle}

\newcommand{\vp}{\textrm{VP}}
\newcommand{\vpws}{\textrm{VP}_{\textrm{ws}}}
\newcommand{\vnp}{\textrm{VNP}}

\newcommand{\barclass}[1]{\overline{#1}}
\newcommand{\vnpbar}{\barclass{\vnp}}
\newcommand{\vpbar}{\barclass{\vp}}
\newcommand{\vpwsbar}{\barclass{\vpws}}
\newcommand{\Cbar}{\barclass{\mathcal{C}}}

\newcommand{\stable}[1]{\textrm{Stable-}#1}
\newcommand{\svp}{\stable{\vp}}
\newcommand{\svnp}{\stable{\vnp}}
\newcommand{\sC}{\stable{\mathcal{C}}}

\newcommand{\svpbar}{\stable{\vpbar}}
\newcommand{\svnpbar}{\stable{\vnpbar}}
\newcommand{\sCbar}{\stable{\Cbar}}

\newcommand{\newton}[1]{\textrm{Newton-}#1}
\newcommand{\nvp}{\newton{\vp}}
\newcommand{\nvnp}{\newton{\vnp}}
\newcommand{\nC}{\newton{\mathcal{C}}}
\newcommand{\nvpws}{\newton{\vpws}}
\newcommand{\nvpbar}{\newton{\vpbar}}
\newcommand{\nvnpbar}{\newton{\vnpbar}}
\newcommand{\nCbar}{\newton{\Cbar}}

\newcommand{\starclass}[1]{#1^{\textrm *}}
\newcommand{\vpstar}{\starclass{\vp}}

\newcommand{\vnpstar}{\starclass{\vnp}}
\newcommand{\Cstar}{\starclass{\mathcal{C}}}
\newcommand{\vpbarstar}{\starclass{\vpbar}}
\newcommand{\vnpbarstar}{\starclass{\vnpbar}}

\newcommand{\C}{\mathbb{C}}
\newcommand{\F}{\mathbb{F}}
\newcommand{\N}{\mathbb{N}}
\newcommand{\R}{\mathbb{R}}
\newcommand{\Z}{\mathbb{Z}}
\newcommand{\Q}{\mathbb{Q}}

\begin{document}
\title{Boundaries  of VP and VNP}
\author{Joshua A. Grochow\footnote{Santa Fe Institute, 
\texttt{jgrochow@santafe.edu}}, Ketan D. Mulmuley\footnote{University of Chicago, 
\texttt{mulmuley@uchicago.edu}}, and Youming Qiao\footnote{University of 
Technology Sydney, \texttt{Youming.Qiao@uts.edu.au}}}



\maketitle

\begin{abstract}
One fundamental question in the context of the geometric complexity theory 
approach to the VP  vs. VNP conjecture is whether VP = $\overline{\textrm{VP}}$, 
where VP is the class  of families of polynomials that are of polynomial degree
and can be 
computed by arithmetic circuits  of polynomial size, and 
$\overline{\textrm{VP}}$ is  the class of families of polynomials that are of 
polynomial degree and can be approximated infinitesimally closely 
by arithmetic circuits of polynomial  size.
The goal of this article\footnote{A preliminary version of this paper appears in
\cite{icalp}}
 is to study the conjecture in (Mulmuley, FOCS 2012) that $\overline{\textrm{VP}}$ 
 is not contained in VP.

Towards that end,  we introduce three degenerations of VP (i.e., sets of points in 
$\overline{\textrm{VP}}$), namely
the stable degeneration Stable-VP, the Newton degeneration Newton-VP, and the 
p-definable one-parameter degeneration VP*. We also introduce analogous  
degenerations of VNP.
We show that Stable-VP $\subseteq$ Newton-VP $\subseteq$ VP* $\subseteq$ VNP, and  
Stable-VNP = Newton-VNP = VNP* = VNP. The three notions of degenerations and the 
proof of this result 
shed light on the problem of separating $\overline{\textrm{VP}}$ from VP.

Although we do not yet construct explicit candidates for the polynomial families 
in $\overline{\textrm{VP}}$ \textbackslash VP, we prove results which tell us 
where not to look for such families.
Specifically,  we demonstrate that the 
families in Newton-VP \textbackslash VP based on semi-invariants  of quivers 
would have to be non-generic  by showing 
that, for many finite  quivers (including some wild ones),  any Newton degeneration of 
a generic semi-invariant can be computed by a circuit of polynomial size.
We also show that the Newton degenerations of  perfect 
matching Pfaffians, monotone arithmetic circuits over the reals, and Schur 
polynomials
have polynomial-size circuits. 
\end{abstract}

\section{Introduction}\label{sec:intro}
One fundamental question in the context of the geometric complexity theory (GCT) 
approach (cf. \cite{GCT1,GCT2}, \cite{landsbergmath}, and \cite{GCT5}) 
to the $\vp$  vs. $\vnp$  conjecture in Valiant \cite{Val79} 
is whether $\vp=\vpbar$, where 
$\vp$ is the class  of families of polynomials that are of polynomial degree and 
can be
computed by arithmetic 
circuits 
of polynomial size, $\vnp$ is the class of p-definable families of polynomials, 
and 
$\vpbar$ is the 
class of families of polynomials that are of polynomial degree and can be 
approximated infinitesimally closely by 
arithmetic circuits   of polynomial size. We assume in what follows 
that the circuits are over an algebraically closed field
$\F$.
We call $\vpbar$ the {\em closure} of $\vp$, and $\vpbar \setminus \vp$ the {\em 
boundary} of $\vp$. 
So the question is whether this boundary is  non-empty.
At present, it is not even known if $\vpbar$ is contained in $\vnp$.

The $\vp$ vs. $\vpbar$  question is important   for two reasons.
First, all known algebraic lower bounds for the exact computation of the permanent 
also hold for 
its infinitesimally close approximation.
For example, the known quadratic lower bound for the permanent \cite{Mignon} also 
holds for 
its infinitesimally closely approximation \cite{LR}, and so also the known lower 
bounds in the algebraic depth-three
circuit models \cite{Kayal}; cf. Appendix~B in \cite{GrochowGCTUnifies} for a  
survey 
of the known lower bounds which emphasizes this point.  These lower bounds hold 
because 
some algebraic, polynomial property that is satisfied by the coefficients of the 
polynomials computed by
the circuits in the restricted class under consideration is  not satisfied by the 
coefficients of the permanent.
Since a polynomial property is a closed condition,\footnote{It is defined by the 
vanishing of 
a continuous function, namely, a (meta) polynomial.} the same property is also 
satisfied by the
coefficients of the polynomials that can be approximated infinitesimally 
closely\footnote{This means
the polynomials are the limits of the polynomials computed by the circuits in the 
restricted class 
under consideration.} 
 by   circuits in the restricted 
class under consideration. This is why the same lower bound also holds for 
infinitesimally close approximation.
We expect the same phenomenon  to hold  in the unrestricted algebraic circuit 
model as well.
Hence, it is natural to expect 
that any realistic proof  of the $\vp \not = \vnp$ conjecture will also show
that $\vnp \not \subseteq \vpbar$, as conjectured in \cite{GCT1}.\footnote{Note that if 
$\vnp \not\subseteq \vpbar$ then there exists a polynomial property showing this 
lower bound.}
This is, in fact, the  underlying thesis   of geometric complexity theory stated in \cite{GCT5}.
But, if $\vpbar  \not = \vp$, as conjectured in \cite{GCT5},  this would mean that 
any realistic approach to the $\vp$ vs. $\vnp$ conjecture 
would even have to separate  the permanent from the families in $\vpbar \setminus 
\vp$ with high circuit complexity.\footnote{Although some lower bounds techniques 
in the restricted models do distinguish between different polynomials with high 
circuit complexity (e.g., \cite{razYehudayoff}), we need a better understanding of 
the families in $\vpbar \setminus \vp$ in order to know which techniques in this 
spirit
could even potentially be useful in the setting of the $\vnp$ versus $\vpbar$ 
problem.}

Second, it is shown in \cite{GCT5} that, assuming a stronger form of the $\vnp 
\not \subseteq \vpbar$ conjecture,
the problem NNL (short for Noether's Normalization Lemma) of 
computing the Noether normalization of  explicit varieties 
can be brought down from EXPSPACE, where it is currently, to P, ignoring a 
quasi-prefix. 
The existing EXPSPACE vs. P gap,\footnote{Or, the EXPH vs. P gap, assuming the 
Generalized Riemann 
Hypothesis.} called the geometric complexity theory (GCT) chasm \cite{GCT5}, in 
the complexity
of NNL may be viewed
as the common cause and measure of the difficulty of the fundamental problems in 
geometry (NNL) and complexity theory
(Hardness). If $\vpbar = \vp$, then it follows \cite{GCT5}  that NNL is in PSPACE.
Thus the conjectural inequality  between $\vpbar$ and $\vp$ is the main difficulty 
that needs to be overcome 
to bring NNL from EXPSPACE to PSPACE unconditionally, and 
is the main reason why the  standard techniques in 
complexity theory may  not be expected to work in the context of the 
$\vp \not = \vnp$ conjecture.

The goal of this article is to study
the conjecture in \cite{GCT5} that $\vpbar$ is not contained in $\vp$. 

\subsection{Degenerations of \texorpdfstring{$\vp$}{VP} and 
\texorpdfstring{$\vnp$}{VNP}} \label{sdegen}
Towards that end,  we introduce three  notions of 
degenerations of $\vp$ and $\vnp$; ``degeneration'' is the standard term in 
algebraic geometry for a limit point or infinitesimal approximation. These degenerations are  
subclasses of $\vpbar$ and $\vnpbar$,
respectively; cf. Section~\ref{sdefine} for
formal definitions.

The first notion is that of a stable degeneration. Recall \cite{mumford} that a 
polynomial
$f  \in \F[x_1,\ldots,x_m]$ is called 
{\em stable} with respect to the natural action of $G=\SL(m, \F)$ on 
$\F[x_1,\ldots,x_m]$
 if the $G$-orbit of $f$ is closed in the Zariski topology. If  $\F=\C$, we may 
equivalently say closed in the usual complex topology. Here  $G$ is $\SL(m,\F)$, and not
$\GL(m,\F)$, since the only polynomial in $\F[x_1,\ldots,x_m]$ with a closed orbit with respect 
to the action of $\GL(m,\F)$ is  identically zero. Hence, whenever
we study issues related to stability in this article, 
we  only consider orbits with respect to the $\SL$-action.

We say that a polynomial $f$ is a {\em stable degeneration} of $g \in 
\F[x_1,\ldots,x_m]$  if 
$f$ lies in a  closed $G$-orbit (which is unique \cite{mumford}) in the closure of 
the $G$-orbit of $g$. 
The degeneration is called stable since $f$ in this case is stable.
We say that a polynomial family $\{f_n\}$ is a {\em stable degeneration} of $\{g_n\}$ if
each $f_n$ is a stable degeneration of $g_n$, with respect to the action of $G=\SL(m_n, \F)$, where $m_n$  denotes the number of variables in $f_n$ and $g_n$.
For any class of polynomial families  $\mathcal{C}$, the class $\sC$ is  defined to be the 
class of families of polynomials that are either in $\mathcal{C}$ or are stable 
degenerations thereof.

The second notion is that of a Newton degeneration. 
We say that a polynomial 
$f$ is a {\em Newton degeneration} of $g$ if it is obtained from $g$ by keeping 
only those terms 
whose associated monomial-exponents lie in some  specified face of the Newton 
polytope of $g$. 
We say that a polynomial family $\{f_n\}$ is a {\em Newton  degeneration} of $\{g_n\}$ if
each $f_n$ is a Newton  degeneration of $g_n$. We say that 
$\{f_n\}$ is a {\em linear projection} of $\{g_n\}$ if each $f_n$ is a linear projection of 
$g_n$.\footnote{This means $f_n$ is obtained from $g_n$ by a 
linear (possibly non-homogeneous) change of variables.}
For any class of polynomial families $\mathcal{C}$, the class $\nC$ is  defined to 
be the class of families of polynomials that are 
Newton  degenerations of the polynomial families  in $\mathcal{C}$, or  are
linear projections of such Newton degenerations.\footnote{Taking a Newton 
degeneration and a linear projection need not commute, so the set of Newton 
degenerations alone will not in general be closed under linear projections. 
For example, any polynomial $f$ is a linear projection of a sufficiently large 
determinant, but the Newton degenerations of the determinant only consist of 
polynomials of the form $\det(X')$ where $X'$ is matrix consisting only of 
variables and 0s.}

The third notion, motivated by the notion of p-definability in Valiant 
\cite{Val79}, is that of a p-definable one-parameter 
degeneration.
We say that a family $\{f_n\}$ of 
polynomials  is a {\em p-definable one-parameter degeneration} of a family 
$\{g_n\}$ of polynomials,  if 
$f_n(x)=\lim_{t \rightarrow 0} g_n(x, t)$, where $g_n(x, t)$ is obtained from 
$g_n(x)$ by transforming its variables $x=(x_1,\ldots,x_i,\ldots)$ linearly 
 such that:
(1) the entries of the linear transformation matrix are Laurent polynomials in $t$ 
of possibly exponential degree
(in $n$), and
(2) there exists a small circuit $C_n$ over $\F$ of size polynomial in $n$
such that any coefficient of the Laurent polynomial in any entry of the 
transformation matrix can 
be obtained by evaluating $C_n$ at the indices of that entry and the index of the 
coefficient. It is assumed here that the indices are encoded as lists of 
$0$-$1$ variables, treating $0$ and $1$ as elements of  $\F$.
Thus a p-definable one-parameter degeneration is a one-parameter 
degeneration of 
exponential degree 
that can be encoded by a small circuit.
For any class of polynomial families $\mathcal{C}$, the class $\Cstar$ is  defined to be the class of 
families of polynomials that are 
p-definable one-parameter  degenerations of the  families in $\mathcal{C}$.

The classes $\vpbar$ and  $\vnpbar$ 
are closed under these three types of degenerations 
(cf. Propositions~\ref{psvpbar}, \ref{pnvpbar}, \ref{pbarstar}).
Since we want to compare $\vpbar$ with $\vp$,
and $\vnpbar$ with $\vnp$, we ask how $\vp$ and $\vnp$ behave under these 
degenerations. This is addressed in the following result.

\begin{theorem} \label{tintro1}
\noindent (a)  $\svnp = \nvnp = \vnpstar = \vnp$, and

\noindent (b)  $\svp \subseteq \nvp \subseteq \vpstar \subseteq \vnp$.
\end{theorem} 

An analogue of this result  also holds for $\vpws$, the class of families 
of polynomials that can be
computed by symbolic determinants of polynomial size. 

\subsection{On $\vpstar$ vs. $\vpbar$ and  $\vp$ vs. $\svp$} \label{squestion}
The statement of Theorem~\ref{tintro1}  tells us nothing as to whether any of the inclusions 
in the sequence  $\vp \subseteq \svp \subseteq \nvp \subseteq \vpstar \subseteq \vpbar$ can be 
expected 
to be strict or not.  But its proof, as discussed below,   does shed   light on 
this subject.

Theorem~\ref{tintro1} is proved by combining the Hilbert-Mumford-Kempf criterion 
for stability \cite{kempf} with
the ideas and results in Valiant \cite{Val79}. The  Hilbert-Mumford-Kempf 
criterion \cite{kempf} shows 
that, for  any polynomial  $f$ in the unique closed $G$-orbit in
the $G$-orbit-closure of any $g \in \F[x_1,\ldots,x_m]$, with $G=\SL(m,\F)$, 
there exists a one-parameter subgroup of $G$ that  drives $g$ to $f$.
Furthermore, by Kempf \cite{kempf}, such a subgroup can  be chosen in a canonical 
manner.
As a byproduct of the  proof of Theorem~\ref{tintro1}, we  get a 
complexity-theoretic  form of this
criterion (cf. Theorem~\ref{thmk1}), which shows that such a  one-parameter group 
can be chosen so that
the resulting one-parameter degeneration of any $\{g_n\} \in \vp$ to $\{f_n\} \in 
\svp$  is   p-definable. Here $f_n$ is a stable degeneration of $g_n$ with respect to the action
of $\SL(m_n,\F)$, where $m_n=\poly(n)$ denotes the number of variables in $f_n$ and $g_n$.
Thus the inclusion of $\svp$ in $\vnp$
ultimately depends on the existence of  a $p$-definable 
one parameter degeneration of $\{g_n\}$ to $\{f_n\}$, 
as provided by the Hilbert-Mumford-Kempf criterion.

However, no such $p$-definable  one parameter degeneration scheme 
is known  if $f_n$ is  allowed to be any polynomial   with a  non-closed 
$\SL(m_n,\F)$-orbit in the $\SL(m_n, \F)$-orbit-closure  of $g_n$, or any polynomial in the $\GL(m_n,\F)$-orbit closure of $g_n$,
regardless of whether the  $\SL(m_n,\F)$-orbit of $f_n$ is closed or not.
Here we consider closedness of the orbits in  the $\GL(m_n,\F)$-orbit-closure of $g_n$ with respect to the action of 
$\SL(m_n,\F)$, not $\GL(m_n,\F)$, since, as pointed out in
Section~\ref{sdegen}, closedness with respect to the $\GL$-action is not interesting.
In other words, we consider the  $\SL(m_n,\F)$- as well as  the $\GL(m_n,\F)$-orbit-closure
of $g_n$ as an  affine $G$-variety, with $G=\SL(m_n,\F)$.

In the context of the $\vp$ vs. $\vpbar$ problem,
one has to consider the $\GL(m_n,\F)$-orbit closure of $g_n$,
since  infinitesimally close approximation involves  $\GL$-transformations.
The $\GL(m_n,\F)$-orbit closures can be much harder than the $\SL(m_n,\F)$-orbit-closures.
For example, if $g_n$ is the determinant, its $\SL(m_n,\F)$-orbit  is already closed \cite{GCT1}, and hence, one 
really needs to understand its $\GL(m_n,\F)$-orbit closure.

If  a  $p$-definable one parameter degeneration 
scheme, akin to the Hilbert-Mumford-Kempf 
criterion for stability, exists 
when  $f_n$ is allowed to be any polynomial  in the $\GL(m_n, \F)$-orbit-closure of $g_n$, 
$\{g_n\} \in \vp$, then it  would follow 
that $\vpbar \subseteq \vpstar$,
and in conjunction with Theorem~\ref{tintro1},  that $\vpbar \subseteq \vnp$. 
This is  one  plausible approach to show that $\vpbar \subseteq \vnp$,
if this is true.\footnote{In this case, separating $\vp$ from 
$\vpbar$ would be stronger than separating $\vp$ from $\vnp$.}
If, on the other hand,  no such $p$-definable one parameter degeneration
scheme  exists when $f_n$ is allowed to be any polynomial 
in the $\GL(m_n, \F)$-orbit-closure of $g_n$, $\{g_n\} \in \vp$, 
then it would be a strong  indication that
$\vpbar$ is not contained  in $\vpstar$, and hence, also not in $\vp$. This would open one possible 
route to formally separate $\vpbar$ from $\vp$.

All the evidence at hand  does, in fact, suggest that such a general scheme may not
exist for the following reasons. 
First, as explained in \cite{mumford} in detail, the Hilbert-Mumford-Kempf criterion 
for stability is intimately related 
to, and in fact, goes hand in hand with  another fundamental result in geometric
invariant theory that,
given any finite dimensional $G$-representation, or more generally, an affine $G$-variety $X$,  $G=\SL(m,\F)$, 
the closed $G$-orbits in $X$
are in one-to-one correspondence with the points of the algebraic variety 
$X/G=\mbox{spec}(\F[X]^G)$,
called the categorical quotient.\footnote{By $\mbox{spec}$ here, we really mean, by abuse of notation, $\mbox{max-spec}$.}
By definition, this is the algebraic variety whose coordinate ring
is $\F[X]^G$, the subring of $G$-invariants in the coordinate ring $\F[X]$ of $X$. But the set of all $G$-orbits in $X$ does   not, in general,  have such a 
natural structure of an algebraic variety \cite{mumford}.  This is why
the book \cite{mumford} focuses on closed $G$-orbits in the construction of the various 
moduli spaces in algebraic geometry.

Second,   from the
complexity-theoretic perspective, the algebraic structure of the set of all $G$-orbits in $X$ seems  much 
harder, in general,
than that of the set of  closed $G$-orbits. For example, it is shown in \cite{GCT5} 
that, if $X$ is a  finite dimensional representation of $G$, then
the set of  closed $G$-orbits in $X$ has a (quasi) explicit system of parametrization
(by a small number of algebraic circuits of small size), assuming that
(a) the categorical quotient $X/G$ is explicit (as conjectured in \cite{GCT5} on the basis
of the algorithmic results therein), and (b) the permanent is hard. In contrast,
it may be conjectured that the set of all $G$-orbits in $X$ does not, in general,  have
an explicit or even a small system of parametrization (by algebraic circuits with $+,-,*,/$, 
and  equality-test  gates), 
since Noether's Normalization Lemma, which plays a crucial role in the parametrization of closed $G$-orbits,
applies only to algebraic varieties.
(The division and equality-test gates are needed here, 
since without them, the outputs of the circuits, being constant on all $G$-orbits, will be $G$-invariant polynomials that cannot distinguish a non-closed $G$-orbit from a $G$-orbit in its closure. By a general result in \cite{rosen}, all $G$-orbits in $X$ can be parametrized, in principle, by a finite number of algebraic circuits of finite size
over the coordinates of $X$,  with $+,-,*,/$, and equality-test gates.)
Formally,  the conjecture is that 
there do not exist for every finite dimensional representation $X$ of $G=\SL(m,\F)$,
$\poly(l,m)$ algebraic circuits of $\poly(l,m)$ size
\footnote{Here the size means the total number of nodes in the circuit. There is
no restriction on the bit-lengths of the constants.}, $l=\dim(X)$,  over the coordinates
$x_1,\ldots,x_l$ of $X$, with constants in $\F$  and gates for $+,-,*,/$, and equality-test, 
such that the outputs of these circuits at the coordinates of any two points $v,w \in X$
are identical iff $v$ and $w$ are in the same $G$-orbit. 
(The gates for division and equality-test  
are not needed for parametrization of closed $G$-orbits in \cite{GCT5}.)

A concrete case that illustrates well the difference between closed $G$-orbits and all $G$-orbits 
is when  $X=M_m(\F)^r$, the space of $r$-tuples of $m\times m$ matrices, with the conjugate (adjoint)
 action
of $G=\SL(m,\F)$. In this case it is known unconditionally
that the set of closed $G$-orbits 
in $X$ has a quasi-explicit (i.e., quasi-$\poly(m,r)$-time computable) parametrization when
the characteristic $p$ of $\F$ is not in $[2,\floor{m/2}]$; cf.  \cite{GCT5} and \cite{fs} for charactetristic
zero, and \cite{GCT5} for positive characteristic. 
In contrast, the best known parametrization \cite{friedland} of all $G$-orbits in $M_m(\F)^r$ (allowing
division and equality-test gates in the algebraic circuits) has exponential complexity.
The known algorithm \cite{sergei} for constructing a canonical normal form of a matrix tuple in $M_m(\F)^r$ 
with respect to
the $G$-action also has exponential complexity,
\footnote{This is because the algorithm in \cite{sergei} needs factorization
of univariate polynomials over extension fields of possibly exponential rank over the base field of definition of
the input.} (though the problem of deciding 
if two points in $M_m(\F)^r$ are in the same $G$-orbit is in $P$ \cite{CIK97,BL08,IKS10}).
The exponential complexity of parametrization of all $G$-orbits in $M_m(\F)^r$ 
may be inherent, since the problem of clasifying all $G$-orbits in 
$M_m(\F)^r$ is  wild \cite{DW05}, when $r \ge 2$. Wildness \cite{drozd,BS03}
 is a universality property in representation theory, analogous to \textrm{NP}-completeness.
 The situation gets even wilder 
when $X$ is a general $G$-representation or an affine $G$-variety. For example, it is known  \cite{belitskii} 
that the problem of classifying all $G$-orbits in $\F^m\otimes \F^m \otimes \F^m$ contains, but is not contained
in the wild problem of classifying  all $G$-orbits in  $M^m(\F)^r$. 

In view of such a fundamental difference
between the algebraic structures of the set of closed $G$-orbits
and the set of all $G$-orbits, from the mathematical as well as the complexity-theoretic perspectives, 
it may be conjectured that a $p$-definable one-parameter degeneration from $\{g_n\} \in \vp$ to $\{f_n\}$, with $f_n$ in 
the $\GL(m_n,\F)$- or $\SL(m_n,\F)$-orbit closure of $g_n$, does not always exist 
if the $G$-orbit of $f_n$, with $G=\SL(m_n,\F)$, is not required to be closed. 
If so, this  would be a strong  indication, as pointed out above,  that
$\vpbar$ is not contained  in $\vpstar$, and hence, also not in $\vp$.

The complexity-theoretic form of the Hilbert-Mumford-Kempf criterion proved in this article
(Theorem~\ref{thmk1}) also
provides an exponential (in $n$) upper bound on the degree of the canonical 
Kempf-one-parameter subgroup 
that drives $g_n$ to $f_n$, with $\{g_n\} \in \vp$ and
$\{f_n\} \in \svp$, where $f_n$ is a stable degeneration of $g_n$.
This canonical Kempf-one-parameter subgroup is known to be the 
fastest way to approach a closed orbit
\cite{Kirwan}. If one could prove a polynomial upper bound on this degree, 
then it would  follow that $\svp=\vp$ (cf. Lemma~\ref{lstrassen}).
On the other hand, if a worst-case super-polynomial lower bound on this degree can 
be  
proved, then it would be a strong  indication that $\svp$, and hence $\vpbar$, are 
different from $\vp$. This would open another possible route 
to formally separate $\vpbar$ from  $\vp$.

\subsection{On the problem of explicit construction}
Next we ask if one can construct an {\em explicit} family in $\nvpws$ that can 
reasonably be conjectured to be not in  $\vpws$ or even $\vp$. With this mind,  we 
first  construct an explicit family $\{f_n\}$ of polynomials
that  can be approximated infinitesimally closely 
by  symbolic determinants of  size $\le n$, but conjecturally cannot be
computed exactly by  symbolic determinants of $\Omega(n^{1+\delta})$ size,
for a small enough positive constant $\delta <1$; cf. Section~\ref{stutte}. This 
construction 
follows a suggestion made in \cite[Section 4.2]{GCT1}.
The family $\{f_n\}$ is a Newton degeneration of the family of perfect matching  
Pfaffians of 
graphs.
However, this family $\{f_n\}$ turns out to be in $\vpws$. So this 
idea needs to be extended much further to construct an explicit family in $\nvpws$ that can 
be conjectured to be not in $\vp$.

To see how this may be possible, note that the  perfect matching  
Pfaffians are derived from a semi-invariant of the  
symmetric quiver with two vertices and one arrow.
This suggests that to upgrade the conjectural
$\Omega(n^{1+\delta})$ lower bound to obtain a candidate for a super-polynomial 
lower bound one could 
replace perfect matching  Pfaffians by  appropriate  representation-theoretic 
invariants (but we do not have to confine ourselves to representation-theoretic invariants;
cf. the remark at the end of Section~\ref{snewton}). 
This leads to the second line of investigation, which we now discuss.

\subsection{On Newton degeneration of generic semi-invariants} \label{snewton}
Our next result suggests that these invariants should  be non-generic 
by showing that, for many finite quivers, including some wild ones, 
Newton degeneration of any generic semi-invariant can be computed by
a symbolic determinant of polynomial size.

A quiver $Q=(Q_0,Q_1)$ \cite{DW00,DZ01} is a directed graph (allowing multiple 
edges) with the set 
of vertices $Q_0$ and the set of arrows $Q_1$.
A linear representation $V$ of a quiver associates to each vertex $x \in Q_0$ a 
vector 
space $V^x$, and to each arrow $\alpha \in Q_1$ a linear map $V^\alpha$ 
from $V^{s \alpha}$ to $V^{t \alpha}$, where $s \alpha$ denotes the start (tail) of
$\alpha$ and $t \alpha$ its target (head). The dimension vector of $V$ is 
the tuple of non-negative integers that associates $\dim(V^x)$ to each vertex $x 
\in Q_0$.
Given a dimension vector $\beta \in \N^{|Q_0|}$, let $\mbox{Rep}(Q,\beta)$ denote 
the space of 
all representations of $Q$ with the dimension vector $\beta$. We have the  natural 
action of
$\SL(\beta):=\prod_{x \in Q_0} \SL(\beta(x), \F)$ on $\mbox{Rep}(Q,\beta)$ by 
change 
of basis.
Let $\mbox{SI}(Q,\beta)=\mbox{Rep}(Q,\beta)^{\SL(\beta)}$ denote the ring of 
semi-invariants.
The {\em generic} semi-invariants in this ring (see \cite{DW00}) will be
recalled  in  Section~\ref{sspecial}.

We will be specifically interested in the following well-known types of quivers, 
cf.  \cite{DW05}.
The $m$-Kronecker quiver is the quiver with two vertices and $m$ arrows between 
the 
two vertices with the same direction. It is
wild if $m \ge 3$.  The $k$-subspace 
quiver is the quiver with $k+1$ vertices $\{x_1,\ldots,x_k,y\}$ and
$k$ arrows $(x_1,y), \ldots, (x_k,y)$. It is wild if $k \ge 5$. The A-D-E Dynkin 
quivers (see
Section~\ref{sgenquiver}) 
are the only quivers of finite representation type---this means they have only
finitely
many indecomposable representations. 

The following  result tells us where {\em not} to look 
for explicit candidate families in $\vpbar \setminus \vp$.

\begin{theorem} \label{tintro2}
Let $Q$ be an $m$-Kronecker quiver, or a $k$-subspace quiver, or an A-D-E Dynkin 
quiver.
Then any Newton degeneration of a generic semi-invariant of $Q$ with dimension 
vector $\beta$ and
degree $d$ can be computed by a weakly skew circuit (or equivalently a symbolic 
determinant)
of $\poly(|\beta|,d)$ size, where $|\beta|=\sum_{x \in Q_0} \beta(x)$.
\end{theorem}

The proof strategy for Theorem~\ref{tintro2} is as follows. Define the coefficient 
complexity $\coeff(E)$ of a  set $E$ of integral linear equalities in $\R^m$ as 
the sum of the absolute values of the coefficients of the equalities. Define 
the coefficient complexity of a face of a polytope in $\R^m$ as 
the minimum of $\coeff(E)$, where $E$ ranges over all integral linear equality 
sets that define the face, in conjunction with the
description of the polytope; cf. Section~\ref{scoeff}.

Theorem~\ref{tintro2}  is proved by showing that  the coefficient complexity of 
every face of the
Newton polytope of a generic semi-invariant of any quiver as above is polynomial 
in $|\beta|$ and $d$, though
the number of vertices on a face can be exponential.

In view of this result and its proof, to construct an explicit  family in $\nvpws 
\setminus \vpws$, we should look for
appropriate  {\em non-generic} invariants of representations of finitely generated 
algebras 
whose Newton polytopes have faces with  {\em super-polynomial coefficient 
complexity} and 
{\em super-polynomial number of vertices}.\footnote{Super-polynomial coefficient complexity 
and super-polynomial number of vertices 
do not ensure high circuit complexity of Newton degeneration. These are necessary conditions
that should only be taken as guiding signs.}

Finally, we emphasize that 
we do not have to confine ourselves to $\nvp$ in the search of a specific 
candidate family in $\vpbar \setminus \vp$. We may search within $\vpstar$, or even outside 
$\vpstar$. Indeed, 
it may be easier to identify specific candidate families in 
$\vpbar \setminus \vp$ outside $\vpstar$ than inside $\vpstar$.

\subsection{Organization.}
The rest of this article is organized as 
follows. In 
Section~\ref{sprelim} we cover the preliminaries. 
In Section~\ref{sdefine}, we formally define the three degenerations of $\vp$ and 
$\vnp$.
In Section~\ref{sproofthm1}, we prove Theorem~\ref{tintro1}. 
In Section~\ref{stutte} we construct an  explicit family $\{f_n\}$ that  can be 
approximated infinitesimally closely 
by  symbolic determinants of  size $\le n$, but conjecturally cannot be
computed exactly by  symbolic determinants of $\Omega(n^{1+\delta})$ size, 
for a small enough positive constant $\delta <1$.
In Section~\ref{sspecial}, we prove Theorem~\ref{tintro2}. 
In Section~\ref{sadditional}, we give  additional 
examples of  representation-theoretic symbolic determinants whose
Newton degenerations have small circuits.
All these examples suggest that  explicit families  in $\nvpws \setminus \vpws$
would have to be rather delicate.

\section{Preliminaries} \label{sprelim}
For $n\in \N$, let $[n]:=\{1, \dots, n\}$.
We denote by  $\vecx=(x_1, \dots, x_n)$  a tuple of variables; $\vecx$ may also denote 
$\{x_1, \dots, x_n\}$. Let $\vece=(e_1, \dots, e_n)$ be a 
tuple of nonnegative integers. We usually use $\vece$  as the exponent vector of 
a monomial in $\F[x_1, \dots, x_n]$. Thus, $\vecx^\vece$ denotes the monomial 
with the exponent vector $\vece$. Let $|\vece|:=\sum_{i=1}^n e_i$. 

For a field $\F$, $\chartc(\F)$ denotes the characteristic of $\F$.
Throughout this paper, we assume that $\F$ is algebraically closed. $\sg_n$ denotes the symmetric group consisting of 
permutations of $n$ objects. 

We say that a polynomial $g=g(x_1,\ldots, x_n)$ is a {\em linear projection} of $f=f(y_1,\ldots,y_m)$
if $g$ can be obtained from $f$ by letting $y_j$'s be some (possibly non-homogeneous) linear combinations of $x_i$'s with
coefficients in the base field $\F$.

A family of polynomials $\{f_n\}_{n \in \N}$ is p-bounded if $f_n$ is a polynomial in $\poly(n)$ variables of $\poly(n)$ degree. The  class $\vp$ \cite{Val79} consists of p-bounded polynomial families $\{f_n\}_{n\in\N}$ over $\F$ such that $f_n$ can be computed by an arithmetic circuit over $\F$ of  $\poly(n)$ size. 

\textbf{Convention:} We call a class $\mathcal{C}$ of families of polynomials \emph{standard} if it contains only p-bounded families, and is closed under linear projections.

By a {\em symbolic determinant} of size $m$ over the variables $x_1,\ldots,x_n$, 
we mean 
the determinant of an $m\times m$ matrix, whose each entry is a possibly non-homogeneous linear function of $x_1,\ldots,x_n$ with coefficients in 
the base field $\F$. The class $\vpws$ is the 
class of families of polynomials  that can be computed by weakly skew circuits of polynomial size, or equivalently,
by symbolic determinants of polynomial size \cite{MP08}.

The class $\vnp$ is the class of p-definable families of polynomials \cite{Val79}, that is, those families $(f_n)$ such that $f_n$ has $\poly(n)$ variables and $\poly(n)$ degree, and there exists a family $(g_n(x,y)) \in \vp$ such that $f_n(x) = \sum_{e \in \{0,1\}^{\poly(n)}} g_n(x,e)$. 

The class $\vpbar$ is defined as follows \cite{GCT1,landsbergmath}. 
Over $\F=\C$, we say that  a polynomial family $\{f_n\}_{n\in\N}$ is in $\vpbar$, if there exists 
a family of sequences of polynomials $\{f^{(i)}_n\}_{n\in\N}$ in $\vp$, $i=1, 2, 
\dots$, such that for every $n$, the sequence of polynomials $f^{(i)}_n$, $i=1, 2, 
\dots$, goes infinitesimally close to $f_n$, in the usual complex topology. Here, polynomials are viewed as points in the linear space of polynomials. 
There is a more general definition that works over arbitrary algebraically closed fields---including in positive characteristic---using the Zariski topology. For a direct treatment, see, e.g. \cite[App.~20.6]{BCS97}. The operational version of this definition we use is as follows: $\{f_n(x_1,\dotsc,x_m)\} \in \vpbar$ if there exist polynomials $f_{n,t}(x_1,\dotsc,x_m) \in \vp_{\C((t))}$---$f_{n,t}$ is a polynomial in the $x_i$ whose coefficients are Laurent series in $t$---such that $f_n(x)$ is the coefficient of the term in $f_{n,t}(x)$ of lowest degree in $t$.

The classes  $\vpwsbar$, $\vnpbar$, and $\Cbar$, for any standard class $\mathcal{C}$,
are  defined similarly.

By the {\em determinantal complexity} $\mbox{dc}(f)$  of a polynomial 
$f(x_1,\ldots,x_n)$, we mean the smallest integer $m$ such that
$f$ can be expressed as a symbolic determinant of size $m$ over $x_1,\ldots,x_n$.
By the {\em approximative determinantal complexity} $\overline{\mbox{dc}}(f)$, we 
mean 
the smallest integer $m$ such that
$f$ can be approximated infinitesimally closely by symbolic determinants of size $m$.

Thus the $\vpws \not = \vnp$  conjecture in Valiant \cite{Val79} is equivalent to saying that $\mbox{dc}(\perm_n)$ is not $\poly(n)$,
where $\perm_n$ denotes the permanent of an $n\times n$ variable matrix.
The $\vnp \not \subseteq \vpwsbar$ conjecture in \cite{GCT1} is equivalent to saying that 
$\overline{\mbox{dc}}(\perm_n)$ is not $\poly(n)$.

A priori, it is not at all obvious that $\mbox{dc}$ and $\overline{\mbox{dc}}$ are different complexity measures.
The following two examples should make this clear.

\begin{ex}[Example 9 in \cite{landsberg2}] 
Let $f= x_1^3 + x_2^2 x_3 + x_2 x_4^2$. Then $\mbox{dc}(f) \ge 5$, but $\overline{\mbox{dc}}(f)=3$.
\end{ex}

\begin{ex}[Proposition 3.5.1 in \cite{LR}]
Let $n$ be odd. Given an $n \times n$ complex matrix $M$, let $M_{ss}$ and $M_{s}$ denote its skew-symmetric  and symmetric parts.
Since $n$ is odd, $\det(M_{ss})$=0. Hence, for a variable $t$,  $\det(M_{ss} + t M_s)= t f(M) + O(t^2)$, for some polynomial function
 $f(M)$. Clearly,  $\overline{\mbox{dc}}(f)=n$, since $\det(M_{ss}+t M_s)/t$ goes 
 infinitesimally close to $f(M)$ when $t$ goes to $0$. But $\mbox{dc}(f) > n$.
 \end{ex}

The $\vpws \not = \vpwsbar$ conjecture in \cite{GCT5} is equivalent to saying that there exists a polynomial family $\{f_n\}$ 
such that $\overline{\mbox{dc}}(f_n)=\poly(n)$, but $\mbox{dc}(f_n)$ is not $\poly(n)$.
Instead of this conjecture, we will focus on the $\vp \not = \vpbar$ conjecture in \cite{GCT5}, since 
the considerations for the former conjecture are entirely similar.

A (convex---we will only consider convex ones here) \emph{polytope} is the convex hull in $\R^n$ of a finite set of points. A \emph{face} of a polytope $P$ is the intersection of $P$ with linear halfspace $H = \{v \in \R^n | \ell(v) \geq c\}$ for some linear function $\ell$ and constant $c$ such that $H$ contains no points of the (topological) interior of $P$. Equivalently, a polytope is the intersection of finitely many half-spaces, a half-space $H_{\ell,c} = \{v | \ell(v) \geq c\}$ is tight for $P$ if $P \subseteq H_{\ell,c}$ and $P \nsubseteq H_{\ell,c'}$ for any $c' > c$, and a face of $P$ is the intersection of $P$ with a half-space of the form $H_{-\ell,-c}$ where $H_{\ell,c}$ is tight for $P$.

\section{Degenerations of \texorpdfstring{$\vp$}{VP} and \texorpdfstring{$\vnp$}{VNP}} \label{sdefine}
To understand the relationship between $\vp, \vnp$, and their closures 
$\vpbar$ and $\vnpbar$, 
we now introduce three  degenerations of $\vp$ and $\vnp$.
The considerations for $\vpws$ and $\vpwsbar$ are entirely similar.

\subsection{Stable degeneration}
First we define stable degenerations of $\vp$ and $\vnp$.

Consider the natural action of $G=\SL(n, \F)$ on   $\F[\vecx]=\F[x_1,\ldots,x_n]$ 
that maps 
$f(\vecx)$ to $f(\sigma^{-1} \vecx)$ for any $\sigma \in G$. Following Mumford et 
al. \cite{mumford},  call  $f=f(\vecx) \in \F[\vecx]$ {\em stable} (with respect 
to the $G$-action) 
if the $G$-orbit of $f$ is Zariski-closed. It is known \cite{mumford} that
the closure of the $G$-orbit of any $g \in \F[\vecx]$ contains a unique closed $G$-orbit.
We say that $f$ is a {\em stable degeneration} of $g$ if $f$ lies in the unique 
closed $G$-orbit in the $G$-orbit-closure of
$g$. (If the $G$-orbit of $g$ is already closed then this just means that $f$ lies 
in the $G$-orbit of $g$.) 

We now define the class $\sC$, the stable degeneration of any standard class $\mathcal{C}$,  as follows.
We say that $\{f_n\}_{n \in \N}$ is in $\sC$ if (1) $\{f_n\} \in \mathcal{C}$, or (2)
there exists $\{g_n\}_{n\in\N}$ in $\mathcal{C}$ such that  each $f_n$ is a stable degeneration 
 of $g_n$ with respect to the action of $G=\SL(m_n, \F)$, where $m_n=\poly(n)$ 
 denotes the 
number of variables in $f_n$ and $g_n$.

\begin{prop} 
For any standard class $\mathcal{C}$ (cf. Section~\ref{sprelim}), $\sC \subseteq \Cbar$. In particular, $\svp \subseteq \vpbar$ and  $\svnp \subseteq \vnpbar$. 
\end{prop}

\begin{proof} 
Suppose $\{f_n(x_1,\ldots,x_{m_n})\}$ is in $\sC$. This means there exists a 
family $\{g_n(x_1,\ldots,x_{m_n})\}$ in $\mathcal{C}$ such that, for each $n$, $f_n$ is in 
the $\SL(m_n, \F)$-orbit closure of $g_n$. This means $f_n$ can be approximated 
infinitesimally
closely by polynomials in $\mathcal{C}$, hence $\{f_n\}$ is in $\Cbar$. 
\end{proof} 

\begin{prop} \label{psvpbar}
$\sCbar = \Cbar$, in particular $\svpbar = \vpbar$, and $\svnpbar = \vnpbar$.
\end{prop}

This is a direct consequence of the definitions.

\subsection{Newton degeneration}
Next we define Newton degenerations of $\vp$ and $\vnp$.

Given a polynomial $f\in\F[x_1, \dots, x_n]$, suppose 
$f=\sum_{\vece}\alpha_\vece\vecx^\vece$. We collect the exponent vectors of $f$ 
and form the convex hull of these exponent vectors in $\R^n$. The resulting polytope is 
called 
the \emph{Newton polytope} of $f$, denoted $\npt(f)$. Given  an arbitrary face 
$Q$ of $\npt(f)$, the Newton degeneration of $f$ to $Q$, denoted $f|_Q$, is 
the 
polynomial $\sum_{\vece\in Q}\alpha_\vece\vecx^\vece$.

We now define the class $\nC$, the Newton degeneration of any class $\mathcal{C}$,  as follows: 
$\{f_n\}_{n\in\N}$ is in 
$\nC$, if there exists $\{g_n\}_{n\in\N}$ in $\mathcal{C}$ such that  each $f_n$ is the Newton 
degeneration of $g_n$ to some face of $\npt(g_n)$, or a linear projection of such a Newton degeneration. 

\begin{theorem}  \label{tvpbar}
Let $\mathcal{C}$ be any standard class (cf.  Section~\ref{sprelim}). Then $\nC \subseteq \Cbar$. In particular, $\nvp\subseteq \vpbar$ and $\nvnp \subseteq \vnpbar$.
\end{theorem}

\begin{proof}
Let $\{f_n\}_{n\in\N}$ be in $\nC$, and suppose $f_n\in\F[x_1, \dots, x_{m(n)}]$. Then there exists $\{g_n\}_{n\in\N}\in\mathcal{C}$, such that  $g_n\in\F[x_1, 
\dots, x_m]$, $m=m(n)$, and $f_n=g_n|_Q$,
where $Q$ is a face of $\npt(g_n)$. Suppose the supporting hyperplane of $Q$ is 
defined by $\langle\veca, \vecx\rangle=b$, where $\veca=(a_1, \dots, a_m)$. 
If necessary, by replacing $(\veca, b)$ with $(-\veca, -b)$, we make 
sure that for an arbitrary exponent vector $\vece$ in 
$g_n$, $\langle \veca, \vece\rangle\geq b$. That is, among all exponent vectors, 
exponent vectors on $Q$ achieve the minimum value $b$ in the direction $\veca$.

Now introduce a new variable $t$, and replace $x_i$ with $t^{a_i}x_i$ to 
obtain a polynomial $g'_n(x_1, \dots, x_m, t)=g_n(t^{a_1}x_1, \dots, 
t^{a_m}x_m)\in\F[x_1, \dots, x_m, t]$. By the definition of 
$f_n$, 
$g'_n=t^b\cdot f_n+ \text{higher order terms in }t.$
Therefore, $\{f_n\} \in \Cbar$.
\end{proof}

\begin{remark}
In the above proof, it is important that the Newton degeneration of $g_n$ is 
the coefficient of $t^b$, the lowest order term in $t$, and it is not at all clear how one could possibly access higher order terms in $t$ using any kind of degeneration. Note that higher order terms can be $\vnp$-complete: Form a matrix of variables $(x_{i,j})_{i,j\in[n]}$, and consider the polynomial 
$\prod_{j\in[n]}(t^{n+1} x_{1,j}+t^{(n+1)^2}x_{2,j}+\dots+t^{(n+1)^{i}}x_{i,j}+\dots+t^{(n+1)^n}x_{n,j})$. 
The coefficient of $t^{(n+1)+(n+1)^2+\dots+(n+1)^n}$ is then the permanent of 
$(x_{i,j})_{i,j\in[n]}$. (Essentially the same construction appeared as 
\cite[Prop. 5.3]{burgfactor}.) The seeming impossibility of extracting higher-order terms in $t$ is in line with the expectation that $\vnp \not\subseteq \vpbar$.
\end{remark}

Noting that if $\mathcal{C}$ is closed under linear projections, then so is $\Cbar$, we have:

\begin{cor} \label{pnvpbar}
For any standard class $\mathcal{C}$, $\nCbar = \Cbar$. In particular, $\nvpbar = \vpbar$ and $\nvnpbar = \vnpbar$.
\end{cor}

\subsection{P-definable one-parameter degeneration} \label{sonepara}
Finally, we define p-definable one-parameter degenerations of $\vp$ and $\vnp$. 
We say that a family $\{f_n(x_1,\ldots,x_{m_n})\}$, $m_n=\poly(n)$,  is a {\em one-parameter degeneration}
of $\{g_n(y_1,\ldots,y_{l_n})\}$, $l_n=\poly(n)$,  of {\em exponential degree}, 
if, for some positive integral function $K(n) =O(2^{\poly(n)})$, there exist
$c_n(i,j,k) \in \F$, 
$1 \le i \le l_n$, $0 \le j \le m_n$, $-K(n) \le k \le K(n)$, 
such that
 $f_n=\lim_{t\rightarrow 0} g_n(t)$, where $g_n(t)$  is obtained from $g_n$ by substitutions of the form 
\[ y_i= a^i_0  + \sum_{j=1}^{m_n}  a^i_j  x_j, \quad 1 \le i \le l_n, \text{ where 
} a^i_j = \sum_{k= -K(n)}^{K(n)} c_n(i,j,k) t^k,  \quad 1 \le i \le l_n, \  0 \le 
j \le m_n.\] 
Note that by \cite{burgfactor}, $\vpbar$ consists exactly of those one-parameter degenerations of $\vp$ of exponential degree.

We say that the family $\{f_n(x_1,\ldots,x_{m_n})\}$, $m_n=\poly(n)$,  is a {\em one-parameter degeneration} of 
$\{g_n(y_1,\ldots,y_{l_n})\}$, $l_n=\poly(n)$, 
of {\em polynomial degree} if $K(n)$ above is $O(\poly(n))$ (instead of 
$O(2^{\poly(n)})$).

We say that a family $\{f_n(x_1,\ldots,x_{m_n})\}$, $m_n=\poly(n)$,  is a {\em p-definable  one-parameter degeneration}
of $\{g_n(y_1,\ldots,y_{l_n})\}$, $l_n=\poly(n)$, 
if, for some $K(n) =O(2^{\poly(n)})$, there exists a $\poly(n)$-size   circuit family $\{C_n\}$ over $\F$ 
such that
 $f_n=\lim_{t\rightarrow 0} g_n(t)$, where $g_n(t)$  is obtained from $g_n$ by substitutions of the form 
\[ y_i= a^i_0  + \sum_{j=1}^{m_n}  a^i_j  x_j, \quad 1 \le i \le l_n, \text{ where 
} a^i_j = \sum_{k= -K(n)}^{K(n)} C_n(i,j,k) t^k,  \quad 1 \le i \le l_n, \  0 \le 
j \le m_n.\] 
Here it is assumed that the circuit $C_n$ takes as input  $\lceil \log_2 l_n \rceil + \lceil \log_2 m_n \rceil + \lceil \log_2 (K(n)+1) \rceil$ many 0-1 variables, which are intended to encode three integers $(i,j,k)$ satisfying $1 \leq i \leq l=l_n$, $0 \leq j \leq m=m_n$, and $|k| \leq K(n)$, 
treating $0$ and $1$ as elements of $\F$.

Thus  a p-definable one-parameter degeneration is a  one-parameter degeneration of 
exponential degree  that
can be specified by a circuit of polynomial size.

\begin{remark}
We can generalize the notion of a one-parameter degeneration slightly by
allowing $C_n$ an additional input $b \in \{0,1\}^{a(n)}$, $a(n)=\poly(n)$, and letting 
\[ a^i_j = \sum_{k= -K(n)}^{K(n)} \left(\sum_{b \in \{0,1\}^{a(n)}} C_n(i,j,k,b)\right) t^k,  \quad 1 \le i \le l_n, \  0 \le j \le m_n.\] 
The following results hold for this more general notion also. 
\end{remark}

For any class $\mathcal{C}$ we now define $\Cstar$, called the p-definable one-parameter degeneration of $\mathcal{C}$, as follows.
We say that $\{f_n\} \in \Cstar$ if there exists $\{ g_n \} \in \mathcal{C}$ such that 
$\{f_n\}$ is a p-definable one-parameter degeneration of $\{g_n\}$.

\begin{lemma}  \label{lvnpstar}
For any standard class $\mathcal{C}$ (cf. Section~\ref{sprelim}), $\nC \subseteq \Cstar$. In particular, 
$\nvp \subseteq \vpstar$ and $\nvnp \subseteq \vnpstar$.
\end{lemma}

This follows from the proof of Theorem~\ref{tvpbar}, noting that we may always take the coefficients of a face to have size at most $2^{\poly(n)}$. The following are easy consequences of the definitions:

\begin{prop} \label{pvpbar}
$\vpstar \subseteq \vpbar$, and  $\vnpstar \subseteq \vnpbar$.
\end{prop}

\begin{proof} 
This is immediate from the definitions. For the first statement,
note that, for any $g_n$ with a small circuit 
and any $a \in \F$, 
$g_n(a)$, which is  obtained from  $g_n(t)$ (cf.  Section~\ref{sonepara})  by setting
$t=a$, also has a small circuit. The situation for the second statement is similar. 
\end{proof} 

\begin{prop} \label{pbarstar}
$\vpbarstar = \vpbar$, and $\vnpbarstar = \vnpbar$.
\end{prop}

This is an immediate consequence of the definitions.

\section{\texorpdfstring{$\svnp=\nvnp=\vnpstar=\vnp$}{Stable-VNP = Newton-VNP = VNP* = VNP}} \label{sproofthm1}
We now prove Theorem~\ref{tintro1}, by a circular sequence of inclusions.

\begin{proof}[Proof of Theorem~\ref{tintro1}] Since $\vnp \subseteq \svnp$ by definition, Theorem~\ref{tintro1} (a)  follows from 
the facts that $\svnp \subseteq \nvnp$ (cf. Theorem~\ref{tsvnp} below), 
$\nvnp \subseteq \vnpstar$ (Lemma~\ref{lvnpstar}),  and $\vnpstar \subseteq \vnp$ (cf.  Theorem~\ref{tvnpstar} below).
 
Theorem~\ref{tintro1} (b)  follows from the facts that $\svp \subseteq \nvp$ (cf. Theorem~\ref{tsvnp} below), 
$\nvp \subseteq \vpstar$ (Lemma~\ref{lvnpstar}),  and $\vpstar \subseteq \vnp$ (cf. Corollary~\ref{cvpst} below).
\end{proof}

\begin{theorem} \label{tsvnp}
For any class $\mathcal{C}$ of families of p-bounded polynomials, $\sC \subseteq \nC$. In particular, $\svp \subseteq \nvp$ and $\svnp \subseteq \nvnp$.
\end{theorem}
\begin{proof}
Suppose $\{f_n\} \in \sC$. If $\{f_n\} \in \mathcal{C}$ then there is nothing to show.
Otherwise, there exists 
 $\{g_n\}_{n\in\N}$ in $\mathcal{C}$ such that  each $f_n$ is a stable degeneration 
of $g_n$ with respect to the action of $G=\SL(m_n, \F)$, where $m_n$ denotes the 
number of variables in $f_n$ and $g_n$. 

It suffices to show that $f=f_n(x_1,\ldots,x_m)$, $m=m_n$, is a Newton degeneration of $g=g_n(x_1,\ldots,x_m)$. 
Let $\vecx=(x_1,\ldots,x_m)$. 

By the Hilbert--Mumford--Kempf criterion for stability \cite{kempf}, there exists a one-parameter subgroup $\lambda(t) \subseteq G$
such that $\lim_{t\rightarrow 0} \lambda(t).g = f$. 
Let $T$ be the canonical maximal torus in $G$ such that 
the monomials in $x_i$'s are eigenvectors for the action of $T$.
After  a linear change of coordinates (which is allowed since $\nC$ is closed under linear transformations by definition), we 
can assume that $\lambda(t)$ is contained in $T$. Thus $\lambda(t)=\mbox{diag}(t^{k_1},\ldots,t^{k_m})$
(the diagonal matrix with $t^{k_j}$'s on the diagonal), $k_j \in \Z$, such that $\sum k_j =1$.

It follows that $f$ is the Newton degeneration of $g$ to the face of $\npt(g)$ 
where the linear function
$\sum_j k_j x_j$ achieves the minimum value (which has to be zero). 
\end{proof} 

The following result is subsumed by Theorem~\ref{tvnpstar}; we include its proof here as a warm-up for expository clarity.

\begin{theorem} \label{tnvnp}
$\nvnp \subseteq \vnp$.
\end{theorem}
\begin{proof} 
Suppose $\{f_n\} \in \nvnp$.
If $\{f_n\} \in \vnp$, then there is nothing to show.
Otherwise,  there exists $\{g_n\}_{n\in\N}$ in $\vnp$ such that  each $f_n$ is the Newton 
degeneration of $g_n$ to some face of $\npt(g_n)$, or  a linear projection of such a Newton degeneration. 
Since $\vnp$ is closed under linear projections,  we can assume, without loss of generality, that
$f_n$ is  the Newton 
degeneration of $g_n$ to some face of $\npt(g_n)$.

By Valiant \cite{Val79}, we can assume that $g=g_n(x_1,\ldots,x_m)$, $m=m_n=\poly(n)$, 
is a projection of $\perm(X)$,\footnote{\label{fn:HC}To get the proof to work in characteristic 2 as well, simply use the Hamilton cycle polynomial $HC(X) = \sum_{\text{k-cycles } \sigma \in S_k} \prod_{i \in [k]} x_{i,\sigma(i)}$ instead, which is $\vnp$-complete in any characteristic \cite{Val79}.} where $X$ is a $k\times k$ variable matrix, with $k=\poly(n)$.
This means $g=\perm(X')$, where each entry of $X'$ is some variable $x_i$ or a constant from the base field $\F$.
Since $f=f_n$ is a Newton degeneration of $g$, it follows that there is some 
substitution, as in the proof of Theorem~\ref{tvpbar}, 
$x_j \rightarrow x_j t^{k_j}$, $k_j \in \Z$, such that
$f=\lim_{t \rightarrow 0} \perm(X'(t))$, where $X'(t)$ denotes the matrix obtained from $X'$ after this substitution.

It is easy to ensure that $|k_j| \leq O(2^{\poly(n)})$.  Then, given any permutation $\sigma \in \sg_k$, 
whether 
the corresponding monomial
$\prod_i X'_{i \sigma(i)}$ contributes to $f$ can be decided in $\poly(n)$ time. It follows that the coefficient of a monomial can be computed by an algebraic circuit summed over polynomially many Boolean inputs (convert the implicit $\poly(n)$-time Turing machine into a Boolean circuit, then convert it into an algebraic circuit (as in \cite[Remark 1]{Val79}) that incorporates the constants appearing in the projection). Hence $\{f_n\} \in \vnp$.
\end{proof}

Since $\vp \subseteq \vnp$, the preceding result implies:

\begin{cor} \label{cnvp}
$\nvp \subseteq \vnp$.
\end{cor}

The following result can proved similarly to Theorem~\ref{tnvnp}.

\begin{theorem} \label{tvnpstar}
$\vnpstar  \subseteq \vnp$.
\end{theorem}

\begin{proof} 
Suppose $\{f_n(x_1,\ldots,x_{m_n})\} \in \vnpstar$.
Then  there exists $\{g_n(y_1,\ldots,y_{l_n})\}$ in $\vnp$ such that  each $f_n$ 
is a p-definable one-parameter degeneration of
$g_n$.

By Valiant \cite{Val79}, we can assume that $g=g_n(y_1,\ldots,y_l)$, 
$l=l_n=\poly(n)$, 
is a projection of $\perm(Y)$,\textsuperscript{\ref{fn:HC}} 
where $Y$ is a $k\times k$ variable matrix, with 
$k=\poly(n)$.
This means $g=\perm(Y')$, where each entry of $Y'$ is some variable $y_i$ or a 
constant from the base field $F$.

Since $f=f_n(x_1,\ldots,x_m)$, $m_n=\poly(n)$, is a p-definable one-parameter 
degeneration of $g$,
 for some $K(n) =O(2^{\poly(n)})$, there exists a $\poly(n)$-size   circuit family 
 $\{C_n\}$ over $\F$ 
such that
 $f_n=\lim_{t\rightarrow 0} g_n(t)$, where $g_n(t)$  is obtained from $g_n$ by 
 substitutions of the form 
\[ y_i= a^i_0  + \sum_{j=1}^{m}  a^i_j  x_j, \quad 1 \le i \le l, \] 
where
\[ a^i_j = \sum_{k= -K(n)}^{K(n)} C_n(i,j,k) t^k,  \quad 1 \le i \le l, \  0 \le j 
\le m.\] 

Let $Y'(t)$ be the matrix obtained from $Y'$ after the substitution above.
Given any permutation $\sigma \in S_k$, and any nonnegative integer sequence 
$\mu=(\mu_0,\ldots, \mu_m)$, 
the coefficient  of the monomial $\vecx^\mu := \prod_i x_i^{\mu_i}$ in  $\prod_i 
Y'(t)_{i,\sigma(i)}$  is a Laurent polynomial in 
$t$. Let $c^\sigma_\mu$ denote the coefficient of $t^0$ in this Laurent 
polynomial. 
It can be  shown that,
for some $\poly(n)$-size circuit $D_n$ over $\F$ (depending on $C_n$)  with 
$s_n=\poly(n)$ inputs, we can express $c^\sigma_\mu$ as 

\[ c^\sigma_\mu = \sum_{b \in \{0,1\}^{s_n}} D_n(b,\sigma,\mu). \]

Here it is assumed that $D_n$ takes $(b,\sigma,\mu)$, specified in binary, as 
input, with $0$ and $1$ regarded as elements of $\F$. The idea is that $\sigma$ specifies which of the $k!$ terms of the permanent is chosen, $\mu$ specifies which monomial is chosen, and the Boolean vector $b$ is used to specify which summand of $y_i$ is chosen in the summation. For a given choice of summand of $y_i$, computing the contribution to the corresponding coefficient is easy using $C_n$, and then we get to sum over all possible choices $b \in \{0,1\}^{s_n}$.

It follows that the coefficient of $\vecx^\mu$ in $f_n$ is

\[ \sum_\sigma c^\sigma_\mu = \sum_\sigma \sum_{b \in \{0,1\}^{s_n}} 
D_n(b,\sigma,\mu). \]

Using this fact, in conjunction with Valiant \cite{Val79}, it can be shown
that $\{f_n\} \in \vnp$. 
\end{proof} 

Since $\vp \subseteq \vnp$, the preceding result implies:

\begin{cor} \label{cvpst}
$\vpstar  \subseteq \vnp$.
\end{cor}

In contrast, using the interpolation technique of Strassen 
\cite{strassen_division} and Bini \cite{bini} we have:

\begin{lemma}[{cf. also \cite{burgfactor}, \cite[\S 9.4]{landsbergmath}, 
\cite[Prop.~3.5.4]{grochowPhD}}] \label{lstrassen} 
If $\{f_n\}$ is a one-parameter degeneration of $\{g_n\} \in \vp$ of polynomial degree, then 
$\{f_n\} \in \vp$.
\end{lemma} 

Theorem~\ref{tintro1} leads to:

\begin{question} \label{qcomplexity}

\noindent (1) Is $\vp=\svp$?

\noindent (2) Is $\vp=\nvp$?

\noindent (3) Is $\vp=\vpstar$?

\noindent (4) Is $\vpstar=\vpbar$?

\end{question}

\subsection{A complexity-theoretic form of the Hilbert--Mumford--Kempf criterion} 
As a byproduct of the proof of Theorem~\ref{tintro1}, 
we get the following complexity-theoretic  form of the Hilbert--Mumford--Kempf criterion \cite{kempf} for stability
with respect to the action of $G=\SL(m, \F)$ on $\F[x_1,\ldots,x_m]$. 
Given a one-parameter subgroup $\lambda(t) \subseteq G$, we can express it as  $A \cdot \mbox{diag}(t^{k_1},\ldots,t^{k_m}) \cdot A^{-1}$,
for some $A \in G$ and $k_j \in \Z$, $1 \le j \le m$. We call $\sum_i |k_i|$ the 
total degree of $\lambda(t)$. The following theorem is implicit in the  proofs of 
Theorems~\ref{tsvnp} and \ref{tnvnp}. 

\begin{theorem} \label{thmk1}
Suppose $f=f(x_1,\ldots,x_m)$ belongs to the unique closed $G$-orbit in the $G$-orbit-closure of $g=g(x_1,\ldots,x_m) \in
\F[x_1,\ldots,x_m]$. Then there exists a one-parameter subgroup $\lambda(t) \subseteq G$ such that (1)
$\lim_{t \rightarrow 0} \lambda(t)\cdot g = f$, and (2) the total degree  of $\lambda$ is $O(\mbox{exp}(m,\bitlength{\deg(g)}))$,
where $\bitlength{\deg(g)}$ denotes the bit-length of the degree of $g$. 

It follows that if $\{f_n\}$ is a stable degeneration of $\{g_n\} \in \vp$, then 
$\{f_n\}$ is a p-definable one-parameter  degeneration of $\{g_n\}$.
\end{theorem}

This result can be generalized from $G=\SL$ to reductive algebraic groups using similar ideas, as follows.
Let $\F=\C$.
Let $V=V_\lambda(R)$ be a finite dimensional rational representation of a 
connected, reductive, algebraic group $R$ with
highest weight $\lambda= \sum_i d_i \omega_i$, where $\omega_i$'s denote the 
fundamental weights of the Lie algebra ${\cal R}$ 
of  $R$. Let $d=\sum_i d_i$. Let $\bitlength{d}$ be its bit-length. Let 
$\mbox{rank}({\cal R})$ denote the rank of ${\cal R}$. 
Given any one-parameter subgroup $\lambda(t) \subseteq R$, let $\hat \lambda: \C 
\rightarrow {\cal R}$ denote the 
corresponding Lie algebra map. After conjugation, we can assume that $\hat 
\lambda(\C)$ is contained in the Cartan subalgebra 
${\cal H} \subseteq {\cal R}$. Fix the standard basis $\{ h_i\}$ of ${\cal H}$ as 
in \cite{fulton}, 
and let $\hat \lambda(1) = \sum_i c_i h_i$. 
Define the total size of $\lambda(t)$ as $\sum_j |c_j|$.

\begin{theorem} 
Given $v \in V$ and $w$ in the unique closed $R$-orbit in the $R$-orbit-closure of 
$v$, there exists a one-parameter
subgroup $\lambda(t) \subseteq R$ of total size  
$O(\mbox{exp}(\rank(R),\bitlength{d}))$   such that 
$\lim_{t \rightarrow 0} \lambda(t) \cdot v =w$. 
\end{theorem}

We formally propose a question that has ramifications on the 
$\svp$ vs. $\vp$ question (cf. Section~\ref{sec:intro} and Question~\ref{qcomplexity} (1)). 

\begin{question}
For some positive constant $a$, does
there exist a stable degeneration $\{f_n\}$ of some $\{g_n\} \in \vp$, 
with an  $\Omega(2^{n^a})$ or a super-polynomial lower bound  on the degree of the canonical Kempf-one-parameter subgroup
\cite{kempf}  $\lambda_n$ driving $\{g_n\}$ to $\{f_n\}$?
\end{question}

\section{Newton degeneration of perfect matching Pfaffians} \label{stutte}
In this section, we 
construct an explicit  family $\{f_n\}$ of polynomials such that
$f_n$ can be approximated infinitesimally closely by  symbolic determinants of size $n$,
but conjecturally requires size $\Omega(n^{1+\delta})$ to be computed by  a symbolic determinant, for a
small enough positive constant $\delta$. However, the family $\{f_n\}$ turns out to be in $\vpws$.

Suppose we have a simple undirected graph $G=(V, E)$ where $V=[n]$. Let $\{x_e\mid 
e\in E\}$ be a set of variables. The Tutte matrix of $G$ is the  $n\times n$ 
skew-symmetric matrix $T_G$ such that,  if $(i,j)=e \in E$, with $i<j$, then 
$T_G(i, j)=x_e$ and $T_G(j,i)=-x_e$; otherwise
$T_G(i,j)=0$. 
For a skew-symmetric matrix $T$, the determinant of $T$ is a perfect square, and 
the square root of $\det(T)$ is called the Pfaffian of $T$, denoted $\pf(T)$. We 
call $\pf(T_G)$ the \emph{perfect matching Pfaffian} of the graph $G$, and 
$
\pf(T_G)=\sum_{P}\sgn(P)\prod_{e\in P} x_e,
$
where the sum is over all perfect matchings $P$ of $G$, and $\sgn(P)$ 
takes $\pm 1$ in a suitable manner. It is 
well-known that $\pf(T_G)\in\vpws$.

Note that $\npt(\pf(T_G))$ is the perfect matching polytope of $G$, which has the 
following description by Edmonds. For any $S\subseteq V$, we use 
$e\sim S$ to denote that $e$ lies at the border of $S$. When $S=\{i\}$, we may  
write $e\sim i$ instead of $e\sim \{i\}$.
\begin{theorem}[Edmonds, \cite{Edm65}]
The perfect matching polytope of a graph $G$ is characterized by the following 
constraints:
\begin{equation}\label{eqn:edmonds}
(a)\, \forall e\in E,  x_e\geq 0; (b) \forall i\in V,  \sum_{e\in E, e\sim i} 
x_e=1; (c)\, \forall C\subseteq V, |C|>1 \text{ is odd},  \sum_{e\in E, e\sim C} 
x_e\geq 1. 
\end{equation}
\end{theorem}
We shall refer to constraints of type (c) in Equation~\ref{eqn:edmonds} as 
``odd-size constraints.''

\begin{theorem}[{Kaltofen and Koiran, \cite[Corollary 
1]{KK08}}]\label{thm:division_ws}
Given $f, g, h \in \F[\vecx]$, suppose $h=f/g$, and $f$ and $g$ are in $\vpws$. 
Then $h\in \vpws$.
\end{theorem}

\begin{theorem}\label{thm:tutte}
For any graph $G$ and any face $Q$ of $\npt(\pf(T_G))$, 
$\pf(T_G)|_Q\in \vpws$.
\end{theorem}
\begin{proof}
Thanks to Edmonds' description, any face of $\npt(\pf(T_G))$ is 
obtained by setting some of the inequalities in Equation~\ref{eqn:edmonds} to 
equalities. As setting $x_e=0$ amounts to consider some graph $G'$ with $e$ 
deleted from $G$, the bottleneck is to deal with the odd-size constraints. 

Suppose the face $Q$ is obtained via setting the odd-size constraints 
corresponding to $C_1$, \dots, $C_s$ to 
equalities, where $C_i\subseteq V$. Note that $s=\poly(n)$, because the dimension 
of $\npt(\pf(T_G))$ is 
polynomially bounded, thus any face can be obtained by setting polynomially many 
constraints to equalities. 
Let $y$ be a new variable. For any edge $e\in E$, let the number of $i\in[s]$ such that
$e$ lies at the border of $C_i$ be $k_e$. Then transform $x_e$ to $x_ey^{k_e}$. 
Let the skew-symmetric matrix after the transformation be $\widetilde{T_G}$. Since 
each perfect matching touches the border of every $C_i$ at least once, $y^s$ 
divides $\pf(\widetilde{T_G})$, so $f:=\frac{\pf(\widetilde{T_G})}{y^s}$ is a 
polynomial. Furthermore, the $y$-free terms in $f$ corresponds to those perfect 
matchings that touch each border exactly once. Thus, 
setting $y$ to zero in $f$ gives $\pf(T_G)|_Q$. 

$f$  is in $\vpws$, because $\pf(\widetilde{T_G})$ and $y^s$ are 
 in $\vpws$, and use 
Theorem~\ref{thm:division_ws}.
\end{proof}

\noindent{\bfseries Construction of an explicit  family.}
Now we turn to the construction of an explicit
family $\{f_n\}$ mentioned in the beginning of this section. We assume that the base field $\F=\C$.

First, we give a randomized procedure for constructing $f_n$: 
\begin{enumerate}
\item Fix a small enough constant $a>0$, and let  $l$ be the nearest odd integer to $n^a$.
Fix  odd-size disjoint subsets $C_1,\ldots,C_k \subseteq [n]$, $k =\floor{n^{1-a}}$, 
of size $l$. For example, we can let  $C_1=\{1,\ldots,l\}$, $C_2=\{l+1,\ldots,2 l +1\}$, etc.

\item Choose a random regular non-bipartite graph $G_n$ on $n$ nodes with degree (say) $\sqrt{n}$.

\item Let $Q$ be the face of $\npt(\det(T_G))$ obtained by setting the odd-size 
constraints corresponding $C_1,\ldots,C_k$
to equalities.

\item Let $f_n=\det(T_G)|_Q$. 
\end{enumerate}

Note that in the above we use determinant instead of Pfaffian, in order to 
simplify the discussion on the determinantal complexity of such polynomials. 
Then,  $f_n$ can be approximated infinitesimally closely by symbolic determinants 
of
size $n$; cf. the proof of Theorem~\ref{tvpbar}. By Theorem~\ref{thm:tutte}, $f_n$ can be expressed as a symbolic 
determinant  of  size $\poly(n)$. But:

\begin{conj}  \label{cexp}
If $a>0$ is
small enough, then,  with a high probability,  $f_n$ cannot be expressed as a symbolic determinant of 
 size $\le n^{1+ \delta}$, for a small enough positive constant $\delta$.
\end{conj} 

This says that the blow-up in the determinantal size in the proof of Theorem~\ref{thm:tutte} due to the use of
division (cf. Theorem~\ref{thm:division_ws}) cannot be gotten rid of completely.

To get an explicit family $\{f_n\}$, we  let $G_n$ be a pseudo-random graph, 
instead of a random graph. This can be done in various ways; perhaps the most conservative way is based on the following result.

\begin{lemma} \label{lpseudo}
Fix a constant $b>0$. Then, the problem of deciding, 
given $G_n$, whether $f_n$ can be expressed as a symbolic determinant of size $\le 
n^b$, belongs to AM.
\end{lemma}
\proof 
This essentially  follows from Theorem~\ref{thm:tutte} and the AM-algorithm for 
Hilbert's Nullstellensatz in Koiran \cite{koiran}, with one additional twist. 

By Theorem~\ref{thm:tutte}. $f_n$ has a small weakly skew circuit $C_n$. Consider 
a generic symbolic determinant $D(x,y)$ of size $n^b$ whose entries are formal 
linear combinations of the $x$-variables, whose coefficients are new $y$ 
variables. We want to know whether there is a setting $\alpha$ of the $y$ 
variables that will make $D(x,\alpha)=C_n(x)$ (note that both the LHS and RHS here 
are given by small weakly skew circuits). The trick on top of Koiran's result is 
as follows.

Randomly guess a hitting set for the $x$ variables---that is, a collection of 
$\poly(n)$ many integral 
values $\xi^{(i)}$ of large enough $\poly(n)$ magnitude 
that will be substituted into the $x$ variables. By the fact that this is a 
hitting set
 (which it is with a high probability  \cite{heintz}), 
the following system of equations has a solution $\alpha$ for the $y$'s iff there 
is a setting of the $y$'s that makes $D(x,\alpha)=C_n(x)$ as polynomials in $x$:
\begin{eqnarray*}
D(\xi^{(1)},y) & = & C_n(\xi^{(1)}) \\
\vdots \\
D(\xi^{(k)},y) & = & C_n(\xi^{(k)})
\end{eqnarray*}
where $k$ is the size of the hitting set. The AM algorithm is then: randomly 
guess the $\xi$'s, then apply Koiran's AM algorithm to the preceding set of 
equations in only the $y$ variables. Note that Koiran's result applies to 
equations given by circuits, and each of the preceding equations is given by a 
small weakly-skew circuit.
\qed

We can now derandomize  the construction of $G_n$ above
using this result in conjunction with the derandomization procedure 
in \cite{KvM} (based on Impagliazzo--Wigderson \cite{IW}), assuming that  E does 
not have Boolean circuits, with an access to the SAT oracle, of subexponential 
size.
This yields, for every $n$, a sequence $G_n^1,\ldots, G_n^l$, $l=\poly(n)$, of 
graphs 
and the corresponding sequence $f_n^1,\ldots,f_n^l$ of polynomials such that
each $f_n^i$ can be approximated infinitesimally closely by symbolic determinants 
of size $n$, but,
assuming Conjecture~\ref{cexp} and the hardness hypothesis above, at least half of 
the
$f_n^i$'s cannot be expressed as symbolic determinants of size $\le n^{1 + 
\delta}$, for a small enough constant $\delta>0$.
This gives a two-parameter explicit family $\{f_n^i | 1 \le i \le l=\poly(n)\}$ 
such that 
each $f_n^i$ can be approximated infinitesimally closely by symbolic determinants 
of size $n$, but,
assuming Conjecture~\ref{cexp} and the hardness hypothesis above, the exact 
determinantal complexity of the family is 
$\Omega(n^{1 + \delta})$, for a small enough constant $\delta>0$.

A less conservative derandomization procedure  is as follows.
For each $n$, let $F_n$ be a Ramanujan graph as in \cite{sarnak} on $n^{3/2}$ 
vertices. 
Partition the set of vertices of $F_n$ in $n$ groups $A_1,\ldots,A_n$, each of 
size $n^{1/2}$.
Let $G_n$ be a graph that contains one vertex labelled $i$ for each $A_i$, $1 \le 
i \le n$.
Join two distinct vertices $i$ and $j$ in $G_n$ if there is an edge in $F_n$ from 
any vertex in  $A_i$ to any vertex in $A_j$.
Let $f_n$ be defined as above with this $G_n$. Then each $f_n$ can be approximated 
infinitesimally closely by 
symbolic determinants of size $n$. But it may be conjectured that $f_n$ cannot be 
computed exactly by a symbolic determinant
of  $\Omega(n^{1+\delta})$ size, for a small enough positive constant $\delta$.

\section{Newton degenerations of generic semi-invariants of quivers} \label{sspecial}
In this section we prove Theorem~\ref{tintro2} for the generalized Kronecker 
quivers, $k$-subspace quivers, and the A-D-E Dynkin quivers.
We assume familiarity with the basic notions of the representation theory of 
quivers; cf. \cite{DW00,DZ01}.

\subsection{Newton degeneration to faces with small coefficient complexity} \label{scoeff}
We begin by observing that the technique used to prove Theorem~\ref{thm:tutte} can be generalized  
further. In the proof of Theorem~\ref{thm:tutte}, due to Edmonds' 
description of the perfect matching polytope, every face has a ``small'' 
description, by  a set of linear equalities whose coefficients are 
polynomially bounded in  magnitude. 

For a face $Q$ of a polytope $P$, we say that a set of linear 
equalities $E$ characterizes $Q$ with respect to $P$, if the description of $P$ together with that of 
$E$ characterizes $Q$. For $E$, let $\coeff(E)$ be the sum of the absolute values of 
the coefficients of the linear equalities in $E$. We define the \emph{coefficient complexity}
 of $Q$ as the minimum of $\coeff(E)$ over the integral linear equality sets $E$ that 
characterize  $Q$ with respect to $P$. Adapting the proof of 
Theorem~\ref{thm:tutte} we easily get the following:

\begin{theorem}\label{thm:small_coeff_complexity}
Suppose $f\in \F[x_1, \dots, x_n]$ can be computed by a (weakly skew) 
arithmetic circuit of size $s$. Let $Q$ be a face of $\npt(f)$ whose coefficient 
complexity is $\poly(n)$. Then $f|_Q$ can be computed by a (weakly skew)  arithmetic 
circuit of size $\poly(s, n)$.
\end{theorem}

\begin{proof}
Let $E=\{\ell_1, \dots, \ell_m\}$ be the set of inequalities characterizing $Q$ 
with respect to $P$, with 
$\coeff(E)$ 
polynomially bounded. Note that  $m=|E|$ is polynomially bounded as well, since 
the 
Newton 
polytope of $f$ 
lives in $\R^n$. Without loss of generality, we assume $\ell_i$ is in the form $a_{i,1}x_1+\dots + 
a_{i,n}x_n\geq a_{i,0}$. Introduce a new variable $y$. For every $j\in[n]$, 
multiply $x_j$ with $y^{\sum_{i\in[m]}a_{i,j}}$. Let $f'\in \F[x_1, \dots, x_n]$ 
be the polynomial obtained from $f$ after transforming each $x_i$ as above. 
Consider 
$f'':=f'/y^{\sum_{i\in[m]}a_{i,0}}$; note that $f''$ is a polynomial. Setting 
$y=0$ in $f''$ yields $f|_Q$. Since all the exponents are polynomially bounded, 
$f'$ (and also 
the polynomial obtained from it
by setting some of the variables to $0$) has a (weakly-skew) 
arithmetic circuit of size $\poly(s, n)$, by Strassen 
\cite{strassen_division} (and Kaltofen and Koiran \cite{KK08}). 
\end{proof}

\begin{remark}
If $Q$ has $\poly(n)$ coefficient complexity, then it can be shown 
that $f|_Q$ is a one-parameter degeneration of $f$ of $\poly(n)$ degree.
Hence, Theorem~\ref{thm:small_coeff_complexity} can also be deduced from Lemma~\ref{lstrassen}.
\end{remark}

\subsection{Generic semi-invariants of generalized Kronecker quivers} \label{sec:Kronecker}
We now prove Theorem~\ref{tintro2} for the $m$-Kronecker quiver. 
Recall that the  $m$-Kronecker quiver is the graph with two vertices $s$ and $t$, 
with $m$ arrows 
pointing from $s$ to $t$. When $m=2$, this is the classical Kronecker quiver. When $m\ge 3$, this quiver is wild.  

Any tuple of $m$ $n\times n$ matrices is  a linear 
representation of the $m$-Kronecker quiver of dimension vector $(n, n)$. 
Let 
$\F[x_{i,j}^{(k)}]$ denote the ring of polynomials in   the  variables $x_{i,j}^{(k)}$,
where $i, j\in[n]$, and $k\in [m]$. For $k \in [m]$, let 
$X_k=(x_{i,j}^{(k)})$ denote the variable $n\times n$ matrix, whose $(i,j)$-th entry is
$x_{i,j}^{(k)}$.
Let $R(n, m)$ consist of those polynomials in 
$\F[x_{i,j}^{(k)}]$ 
that are invariant under the action of every $(A, C)\in \SL(n, \F)\times \SL(n, \F)$, which 
sends $(X_1, \dots, X_m)$ to $(AX_1C^{-1}, \dots, AX_mC^{-1})$. $R(n,m)$ is the ring of semi-invariants for the $m$-Kronecker quiver for dimension vector $(n,n)$ or ``\emph{matrix semi-invariants}'' due to their similarity with the well-known matrix invariants (see Section~\ref{subsubsec:trace}).

\begin{theorem} \label{thm:kronecker}
The Newton degeneration of a generic semi-invariant of the $m$-Kronecker quiver with
dimension vector $(n, n)$ and degree $dn$ to an arbitrary face can be computed by 
a weakly skew arithmetic circuit of size $\poly(d, n)$. 
\end{theorem}

\begin{proof}

Let $M(d,\F)$ be the space of $d \times d$ matrices over $\F$.
By the first fundamental theorem of matrix semi-invariants 
\cite{DW00,SV01,DZ01,ANK07}, $\forall A_1, 
\dots, A_m \in M(d, \F)$, $\det(A_1\otimes X_1+\dots+A_m\otimes X_m)$ is a matrix 
semi-invariant, and every matrix semi-invariant is a linear combination of such 
semi-invariants. When $A_i$'s are generic,  the  monomials occurring
in $\det(A_1\otimes X_1+\dots+A_m\otimes X_m)$ have the following combinatorial 
description \cite{ANK07}. Define a magic square with  the parameter $(n, m, d)$ to be an $n\times 
n$ matrix $S$, with  $(i, j)$-th entry $S(i, j)=(s_{i,j}^{(1)}, \dots, 
s_{i,j}^{(m)})\in \N^m$, satisfying: (1) $\forall i\in [n]$, 
$\sum_{j,k}s_{i,j}^{(k)}=d$, and  (2) $\forall j\in[n]$, $\sum_{i,k}s_{i,j}^{(k)}=d$. 
With such a magic square, we associate  a monomial in $\F[x_{i,j}^{(k)}]$ by
setting the exponent of $x_{i,j}^{(k)}$ to  $s_{i,j}^{(k)}$. 
When $A_i$'s are generic,  the  monomials occurring in $\det(A_1\otimes X_1+\dots+A_m\otimes X_m)$
are precisely the monomials associated with such magic squares.

Consider 
the $n\times n$ complete bipartite graph $G$ in which, for  every $(i, j)\in 
[n]\times [n]$, there are $m$ edges between $i$ and $j$, colored by the elements of the  set $[m]$, with each color 
used exactly once. It is easily seen that the magic squares above correspond to the 
$d$-matchings in this graph $G$: for a graph $G=(V, E)$, a 
$d$-matching is a function $f:E\to \N$ such that  $\forall v\in V$, $\sum_{e\in E, e\sim 
v} f(e)=d$.

Hence,  the Newton polytope of a \emph{generic} matrix semi-invariant is 
characterized by the following constraints:
\begin{equation}\label{eqn:generic_semi-invariants}
(a)\, \forall i, j\in[n], k\in[m], s_{i,j}^{(k)}\geq 0; (b)\, \forall i\in [n], 
\sum_{j,k}s_{i,j}^{(k)}=d; (c)\, \forall j\in [n], \sum_{i,k}s_{i,j}^{(k)}=d.
\end{equation}
This description follows easily from the fact that the incidence 
matrix of a 
bipartite graph (possibly with multiple edges) is unimodular (cf. e.g. \cite[Chap. 
18]{Sch03}).
Each face of this polytope is obtained by setting some of $s_{i,j}^{(k)}$'s to $0$. Hence, 
its coefficient complexity is polynomial in $d$ and $n$.
Therefore, by 
Theorem~\ref{thm:small_coeff_complexity}, the theorem follows.
\end{proof}

\subsection{Generic semi-invariants of \texorpdfstring{$k$}{k}-subspace quivers} 
\label{sksubspace}
Next, we prove Theorem~\ref{tintro2} for the $k$-subspace quivers.

The $k$-subspace quiver is the quiver with $k+1$ vertices $\{x_1, \dots, x_k, 
y\}$, and $k$ arrows $\{\alpha_i=(x_i, y)\mid i\in[k]\}$. For $k=1, 2, 3$, the 
$k$-subspace 
quiver is of finite type. When $k=4$, it is of tame type. When $k\geq 5$, it is 
wild. 

We shall apply the description of semi-invariants of quivers by Domokos and Zubkov 
\cite{DZ01} to the case of $k$-subspace quivers. For this, we need some further 
notions. Fix a field $\F$. Let 
$Q=(Q_0, Q_1)$ be a quiver, where $Q_0$ is the vertex set, and $Q_1$ 
is the arrow set. For an arrow $\alpha$ in $Q$, we use $s\alpha$ (resp. $t\alpha$) 
to denote 
the start (resp. target) of $\alpha$. A path $\pi$ is a sequence of arrows 
$\alpha_1\alpha_2\dots \alpha_\ell$ 
such that $t\alpha_i=s\alpha_{i+1}$ for $i\in[\ell-1]$. The start (resp. target) of 
$\pi$ is $s\alpha_1$ 
(resp. $t\alpha_\ell$). A path is cyclic if $s\pi=t\pi$. We assume $Q$ has no 
cyclic 
paths of positive length. 

Let $V$ be a representation of $Q$; that is, for $x\in Q_0$, 
$V^x$ is the vector space associated with $x$, and for $\alpha\in Q_1$, $V^\alpha$ 
is the 
linear map from $V^{s\alpha}$ to $V^{t\alpha}$. This extends naturally to 
$V^\pi=V^{\alpha_k}\cdots V^{\alpha_1}: V^{s\pi} \to V^{t\pi}$ for a path $\pi$.

Fix a dimension vector $\beta$ for $Q$, and suppose $\ell=|Q_0|$. 
$|\beta|:=\sum_{x\in Q_0} \beta(x)$. Given $\beta$, 
after fixing bases for $V^x$, $x\in Q_0$, a representation of $Q$ is then 
specified using $n:=\sum_{\alpha\in Q_1}\beta(s\alpha)\cdot \beta(t\alpha)$ 
numbers. Let  
$u_1, \dots, u_n$ be $n$ variables. 

Let 
$\GL(\beta):=\GL(\beta_1, 
\F)\times\dots\times \GL(\beta_{\ell}, \F)$ be the direct product of general 
linear 
groups with corresponding dimensions acting naturally on the representations of 
$Q$ with dimension vector $\beta$. Let $SI(Q, \beta)\subseteq \F[u_1, \dots, u_n]$ 
be the set of 
semi-invariants with respect to $Q$ and $\beta$. Any $\sigma:Q_0\to \Z$ defines a 
multiplicative character of $\GL(\beta)$, $\chi_\sigma:(B(x)\mid x\in Q_0)\in 
\GL(\beta) \to \prod_{i\in[\ell]} \det(B(x))^{\sigma(x)}$. Then define $SI(Q, 
\beta)_\sigma=\{f\in SI(Q, \beta) | \forall B\in \GL(\beta), B\cdot f = 
\chi_{\sigma}(B) f\}$. It is clear that $SI(Q, \beta)=\oplus_\sigma SI(Q, 
\beta)_\sigma$. If $\langle \sigma, \beta\rangle:=\sum_{x\in 
Q_0}\sigma(x)\beta(x)\neq 0$ then there are no non-trivial $SI(Q, 
\beta)_{\sigma}$ (see e.g. \cite{DW00}). 
Otherwise, let $\sigma=\sigma_+-\sigma_-$, where $\sigma_+(x)=\max(\sigma(x), 0)$ 
and $\sigma_-(x)=\max(-\sigma(x), 0)$, and set $s=\langle \beta, 
\sigma_+\rangle$. 

Now we come to the key construction. Consider the $s\times s$ matrix
\begin{equation}\label{eq:domokos-zubkov}
g=\oplus_{x\in Q_0} (V^x)^{\sigma_+(x)}\to \oplus_{x\in Q_0} (V^x)^{\sigma_-(x)}, 
\end{equation} 
where each block, $\hom(V^x, V^y)$, is of the form $w_1V^{\pi_1}+\dots 
+w_rV^{\pi_r}$ 
where $\pi_1, \dots, \pi_r$ runs over the set of paths from $x$ to $y$, and $w_1, 
\dots, w_r$ are variables. For different blocks we use different variables. That 
is, the total number of variables is $m=\sum_{x\in Q_0}\sum_{y\in 
Q_0}\sigma_+(x)p(x,y)\sigma_-(y)$, where $p(x, y)$ is the number of paths between 
$x$ and $y$. $\det(g)$ then is a polynomial in $w_1, \dots, w_m$, and $u_1, \dots, 
u_n$. 
$w_i$'s are 
called auxiliary variables, since we shall use the following construction: for 
$(c_1, 
\dots, c_m)\in \F^m$, let $\det(g\mid 
w_i=c_i, i\in[m])$ be the polynomial in $\F[u_1, \dots, u_n]$ after assigning 
$w_i$ with $c_i$ in $\det(g)$. 

\begin{theorem}[Domokos and Zubkov \cite{DZ01}]\label{thm:domokos_zubkov}
Let notations be as above. 
$SI(Q, \beta)_\sigma$ is 
linearly spanned by $\{\det(g\mid w_i=c_i, i\in[m])\mid (c_1, \dots, c_m)\in 
\F^m\}$. 
\end{theorem}

\begin{prop}
The Newton degeneration of a generic semi-invariant of the $k$-subspace quiver of 
dimension vector $\beta$ and degree $d$ to an arbitrary face can be computed by a 
weakly-skew arithmetic circuit of size $\poly(|\beta|, d)$. 
\end{prop}
The proof strategy is to apply Theorem~\ref{thm:domokos_zubkov} to the 
$k$-subspace quiver, which yields a combinatorial description of the exponent 
vectors of monomials in 
a generic semi-invariant of a certain weight. From the combinatorial description 
we obtain a description of the Newton polytope and its faces of a 
generic semi-invariant. We 
then conclude by applying Theorem~\ref{thm:small_coeff_complexity} as for 
generalized Kronecker quivers. 
\begin{proof}
In the $k$-subspace quiver we have $k+1$ vertices $\{x_1, \dots, x_k, 
y\}$ and $k$ arrows $\alpha_i=(x_i, y)$. Observe that (1) a non-trivial path is of 
length 
$1$; (2) only $y$ (resp. $x_i$'s) can serve as the target (resp. start) of a path.
Therefore, for $\det(g)$ to be nonzero, it is necessary 
that $\sigma_+(y)=0$ and $\sigma_-(x_i)=0$ for $i\in[k]$. That is, 
Equation~\ref{eq:domokos-zubkov} 
for $k$-subspace quiver has to be of the form
$$
g=\oplus_{i\in[k]} (V^{x_i})^{\sigma_+(x_i)}\to (V^y)^{\sigma_-(y)}, 
$$
for $\det(g)$ to be nonzero. 

$g$ then is a block matrix of the following form: the rows are divided into 
$\sigma_-(y)$ blocks, with each block of size $\beta(y)$. The columns are divided 
into $\sum_{i\in[k]}\sigma_+(x_i)$ blocks, with $\sigma_+(x_i)$ blocks of size 
$\beta(x_i)$. Let the number of rectangular blocks be $m$ ($m=\sigma_-(y)\cdot 
(\sum_i\sigma_+(x_i))$). As for different blocks we use 
different variables, so the auxiliary variables are $w_1, \dots, w_m$. In a block 
indexed by $(y, x_i)$, we put in $w_\ell V^{\alpha_i}$, where $V^{\alpha_i}$ is a 
variable 
matrix 
of size $\beta(y)\times \beta(x_i)$. We fix bases for $V^y$ and $V^{x_i}$: let 
$P=\{p_1, \dots, p_{\beta(y)}\}$ be a basis of $V^y$, and for $i\in[k]$, let 
$Q_i=\{q_{i, 1}, \dots, q_{i, \beta(x_i)}\}$ be a basis of $V^{x_i}$. Then the 
rows 
of $g$ can be indexed by $\cup_{i\in[\sigma_-(y)]}P^{(i)}$ where $P^{(i)}$ is a  
copy of $P$, and the columns of $g$ are indexed by 
$\bigcup_{i\in[k]}\cup_{j\in[\sigma_+(x_i)]}Q_i^{(j)}$ where $Q_i^{(j)}$ is a copy 
of 
$Q_i$.

Viewing $\det(g)$ as a polynomial in $\F[w_1, \dots, w_m][\cup_{i\in[k]} 
V^{\alpha_i}]$, 
we are interested in the $\cup_{i\in[k]} V^{\alpha_i}$-monomials, as in a generic 
semi-invariant, these are all the monomials. 
Note that $g$ is a $d\times d$ matrix where $d=\beta(y)\cdot \sigma_-(y)$, and 
the $\cup_{i\in[k]} V^{\alpha_i}$-monomials in $g$ are of degree $d$. 

To describe these monomials, we form an undirected bipartite graph $G=(L\cup R, 
E)$ 
as follows. Let $L=\{p_1, \dots, p_{\beta(y)}\}$, and $R=\{q_{i, j_i}\mid i\in[k], 
j_i\in[\beta(x_i)]\}$. Connect each $(p_i, q_{j, \ell})$ with an edge to form a 
complete bipartite graph. The number of edges is 
$\sum_{i\in[k]}\beta(y)\cdot \beta(x_i)$, and we can naturally identify the 
edges with variables in $\cup_{i\in[k]}V^{\alpha_i}$. Now consider a function 
$f:L\cup 
R\to \N$, with $f(p_i)=\sigma_-(y)$, $f(q_{j, \ell})=\sigma_+(x_j)$. For such a 
function, we can define the $f$-perfect matching of $G$, that is a function 
$h:E\to \N$, such that for any $r\in L\cup R$, $\sum_{e\in E, e \sim r}h(e)=f(r)$, 
where $e\sim r$ denotes that $e$ is an edge adjacent to $r$. 

It is not hard to verify that the $f$-perfect matchings and the exponent vectors 
in a generic semi-invariant are in one to one correspondence. One direction is 
easy: a bijective function $b: \cup_{i\in[\sigma_-(y)]}P^{(i)} \to 
\bigcup_{i\in[k]}\cup_{j\in[\sigma_+(x_i)]}Q_i^{(j)}$ clearly defines an 
$f$-perfect matching, as there are $\sigma_-(y)$ copies of $P$ so each $p_i\in P$ 
indexes $\sigma_-(y)$ rows, and similarly for the columns. Furthermore the 
$f$-perfect matching records the exponent vector of the monomial in 
$\cup_{i\in[k]}V^{\alpha_i}$ based on $b$. On the other hand, 
given any $f$-perfect matching, it is routine to check that we can construct at 
least one bijective functions from $\cup_{i\in[\sigma_-(y)]}P^{(i)}$ 
to 
$\bigcup_{i\in[k]}\cup_{j\in[\sigma_+(x_i)]}Q_i^{(j)}$. When there are more than 
one such bijective functions, it is easy to check that all of them produce 
the same monomial. 

This suggests that we can use the description of the bipartite $f$-perfect 
matching polytope. Let $s_{i, j, \ell}$ be the variable associated with the edge 
$(p_i, q_{j, \ell})$, then we have the following inequalities and equalities for 
the bipartite $f$-perfect matching polytope: 
\begin{equation}
\begin{array}{ll}
\forall i\in [\beta(y)], j\in[k], \ell\in[\beta(x_j)], & s_{i,j, \ell}\geq 0  \\
\forall i\in [\beta(y)], & \sum_{j,\ell}s_{i,j, \ell}=\sigma_-(y)\\
\forall j\in [k], \ell\in[\beta(x_j)], & \sum_i s_{i,j, \ell}=\sigma_+(x_j).
\end{array}
\end{equation}
Therefore each face is also obtained by setting some $s_{i, j, \ell}$ to $0$, so 
we can use Theorem~\ref{thm:small_coeff_complexity} to conclude. 
\end{proof}

\subsection{Generic semi-invariants of Dynkin quivers and beyond} 
\label{sgenquiver}
Finally, we prove Theorem~\ref{tintro2} for the A-D-E Dynkin quivers.

Here, instead of the Domokos and Zubkov invariants \cite{DZ01},
we shall use the invariants of Schofield \cite{schofieldSI}, which are also  known 
to linearly span the space of semi-invariants \cite{DW00}. We recall Schofield's 
construction in some detail (without proof), so that we can reason about the 
Newton polytopes of the Schofield invariants. Our general strategy is to find 
(in)equalities satisfied by these Newton polytopes for arbitrary quivers, and then 
to show that for the A-D-E Dynkin  quivers, these inequalities in fact 
\emph{define} the corresponding Newton polytopes. We will then apply 
Theorem~\ref{thm:small_coeff_complexity}, since the inequalities we find all have 
small coefficient complexity.

\subsubsection{Schofield invariants.} 
Given two representations $V,W$ of the same quiver, the associated Schofield 
invariant vanishes if and only if there is a homomorphism of quiver 
representations from $V$ to $W$. The idea is to treat a map from $V$ to $W$ as 
variable, and then try to solve the equations which say that those variables 
define a homomorphism of representations. These equations are linear and 
homogeneous, so they have a solution if and only if a certain determinant 
vanishes; this determinant will be the Schofield invariant.

Consider the map $d^V_W$ that takes $( f^x \colon V^x \to W^x )$ to $( f^{t 
\alpha} \circ V^\alpha - W^\alpha \circ f^{s \alpha} \colon V^{s \alpha} \to W^{t 
\alpha} )$. 
(Here $x$ denotes a vertex of the quiver, $\alpha$ an arrow, $s \alpha$ its start, 
and $t \alpha$ its target.)
For a given $f$, $d^V_W(f)=0$ if and only if $f$ is a homomorphism of quiver 
representations $V \to W$. Thus, $s(V,W) := \det(d^V_W)$ vanishes if and only if 
there \emph{exists} a non-zero homomorphism $V \to W$. Since the existence of a 
homomorphism is basis-independent, we immediately see that the \emph{vanishing} of 
$s(V,W)$ is in fact $\GL(V) \times \GL(W)$-invariant. However, this does not 
immediately tell us that $s(V,W)$ itself is invariant, though it is close; 
Schofield showed that in fact $s(V,W)$ is $\SL(V) \times \SL(W)$-invariant. When 
we fix the dimension vectors of $V$ and $W$, but we think of \emph{both} $V$ and 
$W$ as defined by variables, we call $s(V,W)$ a Schofield \emph{pair} invariant. 
When we think of $V$ as given by actual values but $W$ as given by variables, we 
refer to $s^V(W) = s(V,W)$ as a Schofield invariant. It is these latter 
invariants, as $V$ ranges over all possible dimension vectors and all possible 
values, that linearly span the ring of semi-invariants for the dimension vector of 
$W$ \cite{DW00}.

Let us study the structure of the matrix $d^V_W$ in a bit more detail. For a 
vertex $x$ of a quiver $Q$, we let $V^x$ (resp., $W^x$) denote the vector space 
associated to $x$ in the representation $V$; for an arrow $\alpha$ we let 
$V^\alpha$ denote the corresponding matrix. We use $s \alpha$ to denote the 
``start'' of the arrow $\alpha$ and $t \alpha$ to denote the ``target'' of 
$\alpha$. We think of the matrix as acting on column vectors. The matrix $d^V_W$ 
then has row indices $(\alpha,i,\ell)$, where $\alpha$ is an arrow of $Q$, $i$ 
ranges over a basis for $W^{t \alpha}$ and $\ell$ ranges over a basis for $V^{s 
\alpha}$; it has column indices $(x,j,k)$ where $x$ ranges over vertices, $j$ 
ranges over a basis of $W^x$ and $k$ ranges over a basis of $V^x$. We refer to the 
set of rows of the form $(\alpha,*,*)$ as the $\alpha$ \emph{block-row}, and the 
set of columns of the form $(x,*,*)$ as the $x$ block-column. 

Now we determine the entries of $d^V_W$ precisely:

\begin{eqnarray*}
d^V_W(f)_{\alpha} & = & f^{t \alpha} \circ V^\alpha - W^\alpha \circ f^{s \alpha} 
\\
d^V_W(f)_{(\alpha,i,\ell)} & = & (f^{t \alpha} \circ V^\alpha)_{i,\ell} - 
(W^\alpha \circ f^{s \alpha})_{i,\ell} \\
 & = & \sum_m f^{t \alpha}_{i,m} V^\alpha_{m,\ell} - \sum_n W^\alpha_{i,n} f^{s 
 \alpha}_{n,\ell} \\
d^V_{W; (\alpha,i,\ell),(x,j,k)} & = & \delta_{t \alpha, x} \delta_{i,j} 
V^\alpha_{k,\ell} - \delta_{s \alpha, x} W^\alpha_{i,j} \delta_{\ell,k} \\
 & = & \begin{cases}
\delta_{i,j} V^\alpha_{k,\ell} & \text{ if } t \alpha = x \\
-\delta_{\ell,k} W^\alpha_{i,j} & \text{ if } s \alpha = x \\
\delta_{i,j} V^\alpha_{k,\ell} - \delta_{\ell,k} W^\alpha_{i,j} & \text{ if } s 
\alpha = t \alpha = x \\
0 & \text{ otherwise}
 \end{cases} \\
d^V_{W;\alpha;x} & = & \delta_{t \alpha, x} I_{(W^x)} \otimes (V^\alpha)^T - 
\delta_{s \alpha, x} W^\alpha \otimes I_{(V^x)}
\end{eqnarray*}

Note that each $W^\alpha$ only appears (though multiple times) in a single 
block-row and block-column, namely $(\alpha, s \alpha)$, and each $V^\alpha$ only 
appears in a single block-row and block-column, namely $(\alpha, t \alpha)$. Note 
that every time $V^\alpha$ appears, it appears transposed. From the preceding 
equation, we see that the $\alpha$ block-row is the unique set of $(\dim W^{t 
\alpha})(\dim V^{s \alpha})$ rows that contain all instances of the variables 
$W^\alpha, V^\alpha$, and the $x$ block-column is the unique set of $(\dim 
W^x)(\dim V^x)$ columns containing all instances of the variables $W^\alpha$ for 
all $\alpha$ such that $s \alpha = x$, and all instances of the variables 
$V^\alpha$ for all $\alpha$ such that $t \alpha = x$.

We now begin deriving some inequalities satisfied by $\npt(s(V,W))$. Let 
$\omega^\alpha_{i,j}$ be the exponent corresponding to the variable 
$W^\alpha_{i,j}$ ($\alpha \in E(Q)$, $i \in [\dim W^{t \alpha}]$, $j \in [\dim 
W^{s \alpha}]$), and $\nu^\alpha_{k,\ell}$ the exponent corresponding to the 
variable $V^\alpha_{k,\ell}$ ($k \in [\dim V^{t \alpha}]$, $\ell \in [\dim V^{s 
\alpha}]$). We try to be consistent in our usage of $i,j,k,\ell$ throughout. 

As with all Newton polytopes, we have:
\begin{align}
\omega^\alpha_{i,j} \geq 0 \qquad \nu^\alpha_{k,\ell} &\geq 0 &&  \forall \alpha, 
i, j,k,\ell \label{eqn:pos}
\end{align}

\textit{Blocks.} For each $\alpha \in E(Q)$, in each monomial there must be 
exactly as many entries chosen from the rows in the $\alpha$ block-row as the 
total number of rows in the block-row:
\begin{align}
\sum_{i,j} \omega^{\alpha}_{i,j} + \sum_{k,\ell} \nu^\alpha_{k,\ell} &= (\dim W^{t 
\alpha})(\dim V^{s \alpha}) && \forall \alpha \in E(Q) \label{eqn:blockrow}
\end{align}

Similarly, for each $x \in V(Q)$, in each monomial there must be exactly as many 
entries chosen from the columns in the $x$-th block-column as the total number of 
columns in that block-column:
\begin{align}
\sum_{\alpha : s \alpha = x} \sum_{i,j} \omega^{\alpha}_{i,j} + \sum_{\alpha : t 
\alpha = x} \sum_{k,\ell} \nu^{\alpha}_{k,\ell} &= (\dim W^x)(\dim V^x) && \forall 
x \in V(Q) \label{eqn:blockcol}
\end{align}

\textit{Mini-blocks.} By the \emph{$W$-mini-block-row} corresponding to 
$(\alpha,*,\ell)$ ($\ell \in [\dim V^{s \alpha}]$), we mean the unique set of 
$\dim W^{t \alpha}$ rows in which the $\ell$-th copy of $W^\alpha$ appears. By the 
\emph{$V$-mini-block-row} corresponding to $(\alpha,i,*)$ ($i \in [\dim W^{t 
\alpha}]$), we mean the unique set of $\dim V^{s \alpha}$ rows containing the 
$i$-th copy of $V^\alpha$. Note that the $W$-mini-block-rows and the 
$V$-mini-block-rows appear in a collated or product-like fashion: For each 
$\alpha$, the $(\alpha,i,*)$ $W$-mini-block-row intersects each $V$-mini-block-row 
$(\alpha,*,\ell)$ in exactly one row (namely $(\alpha,i,\ell)$), and vice versa. 
Similarly, by the \emph{$W$-mini-block-column} corresponding to $(x,*,k)$ ($k \in 
[\dim V^x]$), we mean the unique set of $\dim W^{x}$ columns in which the $k$-th 
copy of $W^{\alpha}$ appears for all $\alpha$ such that $s \alpha = x$.  By the 
\emph{$V$-mini-block-column} corresponding to $(x,j,*)$ ($j \in [\dim W^x]$) we 
mean the unique set of $\dim V^x$ columns in which the $j$-th copy of of 
$V^\alpha$ appears, for all $\alpha$ such that $t \alpha = x$.

Consider the $(\alpha,i,*)$ $V$-mini-block-row, i.e., the set of rows in which the 
$i$-th copy of $V^\alpha$ appears, which is the same as the set of rows that 
contain \emph{every} occurrence of the variables $W^\alpha_{i,*}$ (and no other 
$W$ variables). In each copy of $W^\alpha$, each monomial can pick at most one 
variable from the $i$-th row, and there are $\dim V^{s \alpha}$ copies of 
$W^\alpha$ on pairwise disjoint sets of rows ($W$-mini-block-rows), so each 
monomial can have degree at most $\dim V^{s \alpha}$ in the variables 
$W^\alpha_{i,*}$:
\begin{align}
\sum_j \omega^\alpha_{i,j} & \leq  \dim V^{s \alpha} && \forall \alpha \forall i 
\in [\dim W^{t \alpha}]. \label{eqn:Wrow_upper}
\end{align}
On the other hand, if too few elements from $W^\alpha_{i,*}$ are picked, then any 
such monomial will be forced to pick two elements of the $i$-th copy of $V^\alpha$ 
\emph{from the same column}, which is not allowed, so we also have:
\begin{align}
\sum_j \omega^\alpha_{i,j} & \geq \dim V^{s \alpha} - \dim V^{t \alpha} && \forall 
\alpha \forall i \in [\dim W^{t \alpha}] \label{eqn:Wrow_lower_fixed}
\end{align}

Additionally, in the $(\alpha,i,*)$ $V$-mini-block-row, in total each monomial 
must select exactly one variable from each of the $\dim V^{s \alpha}$ rows. The 
natural thing to do here would be to add in the degree of $V^\alpha$. However, in 
doing so we may overcount, since the $\nu$'s also include choices of $V$-variables 
that appear in other rows. So we get:
\begin{align}
\sum_j \omega^\alpha_{i,j} + \sum_{k,\ell} \nu^\alpha_{k,\ell} &\geq \dim V^{s 
\alpha}  && \forall \alpha \forall i \in [\dim W^{t \alpha}]. 
\label{eqn:Wrow_lower}
\end{align}
Despite the similarity of the preceding two inequalities, we note that they are in 
fact independent, as $\sum_{k,\ell} \nu^\alpha_{k,\ell}$ can be larger than $\dim 
V^{t \alpha}$ (e.g., when $\dim W^{t \alpha} > \dim V^{t \alpha}$) or smaller than 
$\dim V^{t \alpha}$ (e.g., when dimensions align so that a term may cover all of 
the $\alpha$ block rows by $W^\alpha$ variables, using none of the $V^\alpha$ 
variables).

For the mini-block-\emph{columns}, we can also get additional information about 
the mini-blocks by considering the ``complement,'' since there is only one 
$W^\alpha$ that is in the same row as any given $V^\alpha$. Given a 
$V$-mini-block-column $(x,j,*)$, we know that exactly $\dim V^x$ entries must get 
chosen from this $V$-mini-block-column. The entries in this $V$-mini-block-column 
are the $j$-th copies of those $V^\alpha$ such that $t \alpha = x$, as well as all 
of the appearances of the columns $W^\alpha_{*,j}$ when $s \alpha = x$. However, 
for each $\alpha$ such that $V^\alpha$ appears in this $V$-mini-block-column, the 
number of $V^\alpha$ entries chosen in this $V$-mini-block-column is complementary 
to the number of $W^\alpha_{j,*}$ entries (same $\alpha$, and the $j$ is 
purposefully the row index now) chosen from the copy of $W^\alpha$ in the $s 
\alpha$ block-column. So we get:
\begin{align}
\sum_{\alpha : s \alpha = x} \sum_i \omega^\alpha_{i,j} + \sum_{\alpha : t \alpha 
= x} \left( \dim V^{s \alpha} - \sum_{j'} \omega^\alpha_{j, j'} \right) &= \dim 
V^x && \forall x \forall j \in [\dim W^x] \label{eqn:Vminiblockcol}
\end{align}

The $V$-$W$ symmetric arguments (with appropriate transposes, etc.) then give:
\begin{align}
\sum_k \nu^\alpha_{k,\ell} & \leq \dim W^{t \alpha} && \forall \alpha \forall \ell 
\in [\dim V^{s \alpha}] \label{eqn:Vrow_upper} \\
\sum_{i,j} \omega^\alpha_{i,j} + \sum_k \nu^\alpha_{k,\ell} & \geq \dim W^{t 
\alpha} && \forall \alpha \forall \ell \in [\dim V^{s \alpha}] 
\label{eqn:Vrow_lower} \\
\dim W^{s \alpha} + \sum_k \nu^\alpha_{k,\ell} & \geq \dim W^{t \alpha} && \forall 
\alpha \forall \ell \in [\dim V^{s \alpha}] \label{eqn:Vrow_lower_fixed} \\
\sum_{\alpha : t \alpha = x} \sum_\ell \nu^\alpha_{k,\ell} + \sum_{\alpha : s 
\alpha = x} \left( \dim W^{t \alpha} - \sum_{k'} \nu^\alpha_{k', k} \right) &= 
\dim W^x && \forall x \forall k \in [\dim V^x] \label{eqn:Wminiblockcol}
\end{align}

In some cases, these equations are already sufficient, for example for the 
generalized Kronecker quivers. These equations may in fact suffice in general 
(with suitable modifications for degenerate cases, such as when certain dimensions 
are 1).

\begin{ex} \label{ex:Kronecker}
As an example, let us observe that these equations reduce to the equations 
(\ref{eqn:generic_semi-invariants}) for the generic semi-invariants of the 
$m$-Kronecker quiver for dimension vector $(n,n)$. In this case, we have two 
vertices, $x, y$, with $m$ arrows $x \to y$, say labelled $1, \dotsc, m$. We have 
$\dim W^x = \dim W^y = n$. There are $n(\dim V^x + \dim V^y)$ columns. Since every 
arrow goes from $x \to y$, every block-row has $(\dim V^x)(\dim W^y)$ rows, for a 
total of $mn(\dim V^x)$ rows. To get these two quantities to be equal, we must 
have $\dim V^x + \dim V^y = m \dim V^x$, or equivalently, $\dim V^y = (m-1) \dim 
V^x$. Let $d = \dim V^x$, so $\dim V^y = d(m-1)$. We then get a Schofield pair 
invariant of total degree $dmn$. However, since all the $W^{\alpha}$ appear in the 
$x$ block-column, which consists of exactly $dn$ columns, the $W$-degree of 
\emph{every} term of $s(V,W)$ is exactly $dn$. This is equivalent to 
(\ref{eqn:blockcol}) for the vertex $x$. This $d$ matches with the $d$ of 
(\ref{eqn:generic_semi-invariants}). 

Considering equation (\ref{eqn:Vminiblockcol}) for the block-column $y$, we find 
that there are no $\alpha$ such that $s \alpha = x$, so the equation reduces to
\begin{eqnarray*}
\sum_\alpha \left( \dim V^{s \alpha} - \sum_j \omega^\alpha_{i,j} \right) & = & 
\dim V^y \qquad (\forall i) \\
\end{eqnarray*}
Since $s \alpha = x$ for all $\alpha$, we have $\dim V^{s \alpha} = \dim V^x = d$ 
for all $\alpha$. Thus $\sum_{\alpha} \dim V^{s \alpha} = md$. Combining this with 
the above equality, and the fact that $\dim V^y = (m-1)d$, yields the second 
equation of (\ref{eqn:generic_semi-invariants}).

Considering equation (\ref{eqn:Vminiblockcol}) for the block-column $x$, we find 
that there are no $\alpha$ such that $t \alpha = x$, and the equation immediately 
becomes the third equation of (\ref{eqn:generic_semi-invariants}). 
\end{ex}

\subsubsection{Dynkin quivers.} The ``A--D--E'' Dynkin quivers are important 
because 
these are the only quivers of \emph{finite} representation type: they have only 
finitely many indecomposable representations. All other quivers have infinitely 
many. The Dynkin quiver of type $A_n$ is defined by having its underlying 
undirected graph be a line on $n$ vertices. $D_n$ is a line on $n-2$ vertices, 
with two additional vertices attached to one end. $E_n$ for $n=6,7,8$ (the only 
ones relevant to the preceding classification) is a path of length $n-1$, together 
with an additional vertex attached to the third vertex on the path. The 
classification statement above is independent of the orientation of the edges, but 
the invariant theory can change with a change in orientation, so we must take some 
care.

\begin{theorem} \label{thm:ADE}
For any of the ADE Dynkin quivers, with arbitrary orientation of arrows, the 
Newton degeneration of a generic semi-invariant with dimension vector $(n_1, 
\dotsc, n_k)$ and degree $d$ to an arbitrary face has determinantal complexity 
$\leq \poly(\sum n_i, d)$.
\end{theorem}

The key to the proof is the following lemma about Schofield \emph{pair} invariants:

\begin{lemma} \label{lem:ADE}
For any of the ADE Dynkin quivers with arbitrary orientation, Equations 
(\ref{eqn:pos})--(\ref{eqn:Wminiblockcol}) define the Newton polytope of any 
Schofield pair invariant $s(V,W)$.
\end{lemma}

Let us first see how the theorem follows from the lemma, then return to prove the 
lemma:

\begin{proof}[Proof of Theorem~\ref{thm:ADE} from Lemma~\ref{lem:ADE}]
The Newton polytope of a \emph{generic} Schofield semi-invariant $s^V(W)$---that 
is, for generic $V$---is the same as the projection of $\npt(s(V,W))$ into the $W$ 
subspace. Let $\pi$ be this projection. The $\pi$-preimage of a face $Q$ of 
$\npt(s^V(W))$ is therefore a face of $\npt(s(V,W))$, and thus $s^V(W)|_{Q} = 
s(V,W)|_{\pi^{-1}(Q)}$. By Lemma~\ref{lem:ADE}, the coefficient complexity of 
$\npt(s(V,W))$ is bounded by a polynomial in $\{\dim V^i, n_i | i \in V(Q)\}$. If 
we can bound these quantities by $\poly(d,n_1,\dotsc,n_k)$, then 
Theorem~\ref{thm:small_coeff_complexity} immediately completes the proof.

To bound the size of $d^V_W$, we determine the $W$-degree of $s(V,W)$ for variable 
$V$. This is easily calculated, using the description of $d^V_W$ above, as 
$$\sum_{x \in V(Q)} \min\{(\dim V^x)(\dim W^x), \sum_{\alpha : s \alpha = x} (\dim 
W^{t \alpha})(\dim V^{s \alpha})\}.$$ 
Thus we see that each component of the dimension vector of $V$ is bounded by $d$ 
(better upper bounds are possible, but this will suffice), which is small enough 
for the preceding argument to go through. Thus we have proved the theorem for 
generic \emph{Schofield} invariants. Finally, as the Schofield invariants $s^V$ 
linearly span the semi-invariants \cite{DW00}, Theorem~\ref{thm:ADE} follows.
\end{proof}

And now we proceed to the proof of the key lemma.

\begin{proof}[Proof of Lemma~\ref{lem:ADE}]
\textit{Type $A_n$.} We proceed by induction on $n$. Note that the $n=1$ case is 
trivial (there are no arrows), and the $n=2$ case is the degenerate case of the 
Kronecker quiver with only a single arrow, which was handled in 
Section~\ref{sec:Kronecker} and Example~\ref{ex:Kronecker}. 

Suppose that one of the end vertices, $x$, is a source (the sink case is 
analogous, swapping the roles of $V$ and $W$ and transposing if needed). Call the 
unique outgoing arrow $\alpha \colon x \to y$. Then in the $(x,*,*)$ block column, 
only $W^\alpha$ appears. Then (\ref{eqn:Vminiblockcol}) gives that $\sum_i 
\omega^\alpha_{i,j} = \dim V^x$ for all $j \in \dim W^x$. Summing over $j$ then 
gives $\sum_{i,j} \omega^\alpha_{i,j} = \dim V^x \dim W^x$. Equation 
(\ref{eqn:blockrow}) then becomes $\dim V^x \dim W^x + \sum_{k,\ell} 
\nu^\alpha_{k,\ell} = \dim W^y \dim V^x$, or equivalently
$$
\sum_{k,\ell} \nu^\alpha_{k,\ell} = \dim V^x ( \dim W^y - \dim W^x).
$$
By the nonnegativity of $\nu^\alpha_{k,\ell}$ we see that there can be no 
nontrivial semi-invariants unless $\dim W^y \geq \dim W^x$; furthermore, if $\dim 
W^y = \dim W^x$, then the Schofield pair invariants do not involve $V^\alpha$ at 
all. It is not hard to see that in this case the Schofield pair invariant is just 
$\det(W^\alpha)^{\dim V^x}$ times the Schofield pair invariant $s_x(V,W)$ of the 
quiver one gets by deleting the vertex $x$. Then $\npt(s(V,W)) = ((\dim V^x) \cdot 
\npt(\det(W^\alpha))) \times \npt(s_x(V,W))$, where the scalar dot here represents 
scaling the polytope (technically this is a Minkowski sum of 
$\npt(\det(W^\alpha))$ with itself $\dim V^x$ times, but since Newton polytopes 
are, in particular, convex, this is the same as scaling up the polytope). The 
product here represents Cartesian product, since the constraints on the $W^\alpha$ 
are independent of the constraints on the remaining variables in this case. So the 
inequalities are the $\dim V^x$-scaled inequalities for $\npt(\det(W^\alpha))$ 
(which is just a rescaling of the perfect matching polytope for the complete 
bipartite graph), and those for $\npt(s_x(V,W))$, which are 
(\ref{eqn:pos})--(\ref{eqn:Wminiblockcol}), by induction.

Otherwise, $\dim W^y > \dim W^x$. In this case, (\ref{eqn:Wminiblockcol}) says 
that $\sum_k \nu^\alpha_{k,\ell} \geq \dim W^y - \dim W^x$. Summing these over all 
$\ell$ and combining with the equation above, we get that in fact $\sum_k 
\nu^\alpha_{k,\ell}$ is \emph{equal} to $\dim W^y - \dim W^x$, for each $\ell$. 

Now, we have two cases: either the other arrow incident on $y$ is oriented towards 
$y$ or away from $y$. 

Case 1: $\beta \colon z \to y$ is oriented towards $y$ (in particular, $y$ is a 
sink). In this case, the only entries that appear in the $y$-block-column are 
$V^\alpha$ and $V^\beta$. The equations for this block-column and the 
corresponding mini-block-columns are then exact, since there is no ``mixing of 
indices'' that would occur had there been both $V$'s and $W$'s in the same column. 
This allows us to easily link to the rest of the $d^V_W$ matrix and proceed 
inductively.

Case 2: $\beta \colon y \to z$ is oriented away from $y$. In this case, the 
entries that appear in the $y$-block-column are $V^\alpha$ and $W^\beta$. In this 
case, the equations for the $y$ block column and the corresponding 
mini-block-columns give constraints on the total degree of the $W^{\beta}$, and 
(\ref{eqn:Wminiblockcol}) links $V^\alpha$ (in the $y$-block column) to $V^\beta$ 
in the $z$-block column by an \emph{equality} (not inequality), again allowing us 
to easily proceed inductively to the rest of the matrix.

\textit{Type $D_n$.} The induction is in fact the same, starting from the ``long 
end,'' of the quiver. The base case, however, is now $D_4$. One orientation of the 
$D_4$ quiver gives the 3-subspace quiver, which was handled in 
Section~\ref{sksubspace}. So we must handle the other possible orientations. Let 
$u$ be the ``central'' vertex (degree 3), and $v_1, v_2, v_3$ the outer vertices 
of degree 1. Because of the $S_3$ symmetries of $D_4$, there are only four 
possible orientations up to symmetry, determined precisely by whether there are 
0,1,2, or 3 arrows pointing towards $u$. When there are 3 we have the 3-subspace 
quiver, so we need only handle the other three cases.

We use $\alpha_i$ to denote the arrow between $v_i$ and $u$ (whichever direction 
it is facing), for each $i=1,2,3$.

Case 0: All of the arrows are pointing \emph{away} from $u$. In this case, the 
$u$-block-column contains only $V$ blocks, and each $v_i$ block column contains 
only its corresponding $W$ blocks. As in the $A_n$ case, this means that $\dim 
W^{u} \geq \dim W^{v_i}$ for all $i=1,2,3$ is required to get any non-constant 
semi-invariants. If $\dim W^{v_i} = \dim W^u$ for some $i$, then $V^{\alpha_i}$ 
doesn't appear at all in the Schofield pair invariant, and the Schofield pair 
invariant is a power of $\det(W^{\alpha_i})$ times the Schofield pair invariant 
for the quiver representation gotten by removing the vertex $v_i$. But the 
remaining quiver in this case is an $A_3$ quiver, which is covered by the $A_n$ 
case above. 

So now we assume that $\dim W^{u} > \dim W^{v_i}$ for all $i$. Any term of 
$s(V,W)$ is therefore determined by: (a) picking exactly $\dim W^{v_i}$ 
entries---in distinct rows and columns---from each of the $\dim V^{v_i}$ 
occurrences of $W^{\alpha_i}$, (b) for each $\ell$, picking exactly $\dim W^u - 
\dim W^{v_i}$ entries from among the $V^{\alpha_i}_{*,\ell}$ that are in rows 
different from those entries picked in (a), and (c) ensuring that for each $k$, 
exactly $\dim W^u$ entries are chosen from among the $V^{\alpha_{*}}_{k,*}$. Our 
goal is to show that any vertex of the polytope defined by 
(\ref{eqn:pos})--(\ref{eqn:Wminiblockcol}) satisfies these conditions. 

Equation (\ref{eqn:Vminiblockcol}) for $(v_1, j, *)$ turns into $\sum_i 
\omega^{\alpha_1}_{i,j} = \dim V^{v_1}$, and similarly for $v_2$ and $v_3$. 
For $(u,j,*)$, (\ref{eqn:Vminiblockcol}) becomes $\sum_{i \in [3]} ( \dim V^{s 
\alpha_i} - \sum_{j'} \omega^{\alpha_i}_{j,j'}) = \dim V^u$. 

Equation (\ref{eqn:Wminiblockcol}) for $(u,*,k)$ turns into $\sum_{i \in [3]} 
\sum_\ell \nu^{\alpha_i}_{k,\ell} = \dim W^u$, which is precisely condition (c). 
For $(v_i, *, k)$ it becomes $\dim W^u - \sum_{k'} \nu^{\alpha_i}_{k',k} = \dim 
W^{v_i}$, or equivalently $\sum_k \nu^{\alpha_i}_{k,\ell} = \dim W^u - \dim 
W^{v_i}$ for all $\ell$ (that is, (\ref{eqn:Vrow_lower_fixed}) is already an 
equality for this quiver). 

For this quiver, (\ref{eqn:Vrow_lower}) is either redundant or automatically an 
equality. By the preceding paragraph, (\ref{eqn:Vrow_lower}) turns into 
$\sum_{i,j} \omega^{\alpha_1}_{i,j} \geq \dim W^{v_1}$, when in fact we already 
have that $\sum_{i,j} \omega^{\alpha_1}_{i,j} = \dim V^{v_1} \dim W^{v_1}$. So if 
$\dim W^{v_1} = 1$, then (\ref{eqn:Vrow_lower}) is automatically an equality, and 
if $\dim W^{v_1} > 1$ then it is redundant so we need not worry about using it for 
defining vertices. Similarly for $v_2, v_3$. 

For this quiver, we also see that (\ref{eqn:Vrow_upper}) is redundant, since 
(\ref{eqn:Vrow_lower_fixed}) is an equality. This leaves only the nonnegativity 
constraints (which are easily satisfied by equality by considering monomials not 
involving the given variable), and (\ref{eqn:Wrow_upper})--(\ref{eqn:Wrow_lower}).

Setting (\ref{eqn:Wrow_upper}) to an equality amounts to picking exactly on 
element from the $i$-th row of each copy of $W^{\alpha}$, and there are certainly 
monomials in the Schofield invariant that do this. Once this is done, 
(\ref{eqn:Wrow_lower}) and (\ref{eqn:Wrow_lower_fixed}) become redundant. 

Setting (\ref{eqn:Wrow_lower_fixed}) to an equality amounts to only picking 
elements of $W^{\alpha}$ from its $i$-th row in $\dim V^{v_1} - \dim V^u$ of the 
copies of $W^{\alpha}$, rather than from each copy. It is possible to find 
monomials that do this for several values $i$; when it is no longer possible, this 
is because there simply aren't enough rows to accommodate not picking from some of 
them. But this is ruled out by the relations needed between the dimensions of the 
$V$'s and $W$'s to make $d^V_W$ square. Once this is done, (\ref{eqn:Wrow_upper}) 
becomes redundant, and (\ref{eqn:Wrow_lower}) becomes $\sum_{k,\ell} 
\nu^\alpha_{k,\ell} \geq \dim V^{v_1}$. 

Setting (\ref{eqn:Wrow_lower}) to an equality can be achieved by considering only 
the terms in $s(V,W)$ where $V^{\alpha}$ is only taken from its $i$-th occurrences 
(equivalently, by zero-ing out all except the $i$-th occurrence of $V^{\alpha}$). 
This can be done as often as we like, so long as there are enough $V$'s left in 
the $u$-block-column. But when there are no longer enough left, we will have 
reached the empty polytope, again by the necessary relations between the 
dimensions of $V$'s and $W$'s.

Thus for any vertex defined by the above equations, there is a monomial in 
$s(V,W)$ with that exponent vector, as desired.

The remaining cases, although they at first seem cosmetically different, turn out 
to have essentially the same proof. The key fact that enables this is that in each 
$v_i$ block column, at most one matrix appears, albeit multiple times.

\textit{Type $E_n$.} The induction is essentially the same, except now we must 
induct starting from \emph{each} of the three ``long ends,'' until we again get 
down to a $D_4$ quiver as the base case. This was handled above, so we are done.
\end{proof}

This proof easily extends to any quiver that is a tree with at most one vertex of 
degree $> 2$. We expect that it should extend without much difficulty to arbitrary 
trees, and it may even extend to completely arbitrary quivers.

\section{Additional examples of Newton degenerations in 
\texorpdfstring{$\vp$}{VP}} \label{sadditional}
In this section, we give  additional examples of representation-theoretic symbolic 
determinants 
whose Newton degenerations can  be computed by symbolic determinants of polynomial 
size.
These examples suggest that 
explicit  families in $\nvpws \setminus \vpws$ 
have to be rather  delicate.

\subsection{Schur polynomials, and the self-replication phenomenon}
For an integer $\ell$, the elementary symmetric polynomial 
$e_\ell(\vecx)\in\Z[x_1, \dots, x_n]$ is $\sum_{1\leq i_1<i_2<\dots<i_\ell\leq 
n}x_{i_1}\cdot x_{i_2}\cdot \cdots \cdot x_{i_\ell}$. In particular, for $\ell>n$ 
or $\ell<0$, $e_\ell(\vecx)=0$; for $\ell=0$, $e_\ell=1$.

A partition $\alpha=(\alpha_1, \dots, 
\alpha_n)$ is an element  in $\N^n$ with $\alpha_1\geq \alpha_2\geq \dots\geq 
\alpha_n$. The conjugate of $\alpha$, denoted $\alpha'$, is a partition in 
$\N^m$ defined by setting $\alpha'_i=|\{j\in[n]\mid \alpha_j\geq i\}|$, where 
$m\geq \alpha_1$. 
Let $|\alpha|:=\sum_i\alpha_i$. 

The Schur polynomial $\schur_\alpha(\vecx)$ in $\Z[\vecx]$ can be defined by 
the Jacobi--Trudi formula as:
$\schur_\alpha(\vecx)=\det[e_{\alpha'_i-i+j}(\vecx)]_{i, 
j\in[n]}$ \cite{fulton}. As $e_\ell$ can be computed by a depth-$3$ arithmetic 
formula of size 
$O(\ell^2)$ (by Ben-Or; see \cite{NW97}), $\schur_\alpha$ can be computed by a 
weakly-skew circuit 
of 
size $\poly(|\alpha|, n)$. 

It is well-known that Schur polynomials are symmetric and homogeneous. For a 
partition $\beta$ such that  $|\beta|=|\alpha|$, the monomial $\vecx^\beta$ is in 
$\schur_\alpha$ if and only if $\alpha$ \emph{dominates} $\beta$; that is, for 
every $i\in[n]$, $\sum_{j=1}^i\alpha_j\geq \sum_{j=1}^i\beta_j$. It follows that 
the Newton polytope of $\schur_\alpha$ is the permutohedron with respect to  
$\alpha$, 
denoted $PH_\alpha$. It is defined as follows. For $\pi\in \sg_n$, let 
$\alpha^\pi=(\alpha_{\pi(1)}, \dots, \alpha_{\pi(n)})$. Then $PH_\alpha$ is the 
convex hull of $\alpha^\pi$'s, where $\pi$ ranges over all permutations in  
$\sg_n$. 

To determine the Newton degeneration of Schur polynomials, we need to understand 
faces of this permutohedron. Any face $Q$ of $PH_\alpha$ is determined by a 
sequence of nested subsets $\emptyset\subsetneq S_1\subsetneq \dots \subsetneq 
S_k=[n]$. For such a sequence, $Q$ is the convex hull of $\alpha^\pi$'s with 
$\pi\in 
\sg_n$ satisfying: for any $1\leq i<j\leq k$, and any $p\in S_i$ and $q\in S_j$, 
$\alpha_{\pi(p)}\leq \alpha_{\pi(q)}$. Thus, the face $Q$ is the Minkowski sum of 
$k$ permutohedra, and $\schur_\alpha|_Q$ is a product of Schur polynomials, which 
in turn can be computed by a weakly-skew arithmetic circuit of size 
$\poly(|\alpha|, n)$. This yields:

\begin{prop}
The Newton degeneration of a Schur polynomial to any face of its Newton polytope 
is a product of some Schur polynomials. 
\end{prop}

Thus the Newton polytope of a Schur polynomial has the \emph{self-replication} 
property: any face of this polytope
is a Minkowski sum of  smaller polytopes of the same kind. 
Another polynomial with such  self-replication property is the resultant, cf. 
Sturmfels 
\cite{Sturmfels}. 

\subsection{Monotone circuits}
An arithmetic circuit over $\Q$ or $\R$ is called \emph{monotone} if it uses only 
nonnegative field elements as constants. No cancellation can happen during the 
computation by
a monotone  circuit. This fact can be used to prove the following result.

\begin{theorem}\label{thm:monotone}
If $f\in\R[x_1, \dots, x_n]$ can be computed by a monotone arithmetic circuit of 
size $s$, then, for any face $Q$ of $\npt(f)$, $f|_Q$ can be computed by a 
monotone 
arithmetic circuit of size $\leq s$.
\end{theorem}
\begin{proof}
Let $C$ be the a monotone arithmetic circuit of size $s$ computing $f$. Since $C$ 
is monotone, each subcircuit is also monotone and computes a polynomial with 
nonnegative coefficients.

Let $Q$ be defined by the supporting hyperplane 
$\langle\veca, \vecx\rangle=b$. Therefore, a monomial $\vecx^\vece$ in $f$ is 
present in $f|_Q$ if and only if $\langle \veca, \vece\rangle$ achieves the 
minimum $b$ among all monomials in $f$. 

We now transform $C$ to a monotone circuit $C'$ computing $f'$. The resulting $C'$ 
will be of the same structure of $C$, except that some edges may be removed. In 
particular, the gate sets of $C$ and $C'$ are the same: for a gate $v$ in $C$, we 
shall use $v'$ to denote the correspondent of $v$ in $C'$. We also use $f_v$ to 
denote the polynomial computed at the gate $v$. The goal is to ensure that, for 
any 
gate $v'$, in $C'$, $f_{v'}$ consists of those monomials achieving the minimum 
along 
$\veca$ among all monomials in $f_v$. 

This is achieved by induction on the depth. As the base case, we keep all the 
leaves of $C$ to be the leaves of $C'$.

Now we proceed to  the inductive step. By the induction hypothesis, for every gate 
$v'$ of depth $\leq d$, $f_{v'}$ consists of those monomials achieving the  
minimum along 
$\veca$ among all monomials in $f_v$. Let $w$ be a gate at depth $d+1$ in $C$, 
with two children $v$ and $u$. Suppose in $f_v$ (resp. $f_u$) the monomials 
achieve the minimum $b_v$ (resp. $b_u$) along $\veca$. If $w$ is labeled with $+$ 
and
 $f_w=f_v+f_u$, then we set $f_{w'}=f_{v'}$ if $b_v<b_u$, $f_{w'}=f_{u'}$ if 
$b_v>b_u$, and $f_{w'}=f_{v'}+f_{u'}$ if $b_u=b_v$. If $w$ is labeled with 
$\times$ and  $f_w=f_v \times  f_u$, then we set $f_{w'}=f_{v'}\times f_{u'}$. 
Since no cancellation can
happen, it follows from the induction hypothesis that
$f_{w'}$ consists of those monomials achieving the minimum along $\veca$ 
among all monomials in $f_w$.
\end{proof}

\begin{remark}
\begin{enumerate}
\item The resulting circuit $C'$ in the proof of Theorem~\ref{thm:monotone}
preserves  most of the structural properties of $C$. In particular, if $C$ is 
weakly-skew,
 then $C'$ is also weakly-skew.
\item The preceding  proof does not  work in the presence of cancellations. For 
example, when $w$ is labelled with $+$, we have $f_w=f_v+f_u$, but $f_{w'}$ may 
not be  $f_{v'}$,  $f_{u'}$, or $f_{v'}+f_{u'}$ due to cancellations. 
\end{enumerate}
\end{remark}

\subsubsection{Trace monomials}\label{subsubsec:trace}
We now apply Theorem~\ref{thm:monotone} to classical matrix invariants.

Let $\F[x_{i,j}^{(k)}]$ be the ring of polynomials in the variables 
$x_{i,j}^{(k)}$, where $i, j\in[n]$, 
and $k\in [m]$. Let
$X_k=(x_{i,j}^{(k)})$ be an $n \times n$ variable matrix whose $(i,j)$-th entry is 
$x_{i,j}^{(k)}$.
The classical matrix invariants are those polynomials in  $\F[x_{i,j}^{(k)}]$ that 
are invariant under the conjugation action 
of $\SL(n, \F)$. Under this action, $A \in \SL(n, \F)$ 
 sends $(X_1, \dots, X_m)$ to $(AX_1A^{-1}, \dots, AX_mA^{-1})$.  
The matrix invariants are semi-invariants of the quiver with a single vertex and 
$m$ self-loops.

Important examples of matrix invariants are the 
polynomials of the form $\Tr(X_{i_1}\cdot \cdots\cdot X_{i_\ell})$,
$i_j\in[m]$, called 
the {\em trace monomials}. By the first fundamental theorem of matrix 
invariants \cite{Pro76},  every matrix invariant is a linear 
combination of trace monomials in  characteristic zero.

Suppose $\F=\Q$ or $\R$. Then any
Newton degeneration of a trace monomial can be computed by a 
weakly skew arithmetic circuit of size $\poly(n, \ell)$. This follows from 
Theorem~\ref{thm:monotone}, since  a trace
monomial can be computed by a monotone weakly skew arithmetic circuit of size 
$\poly(n, \ell)$.

\section*{Acknowledgment}
We thank the anonymous reviewers for careful 
reading 
and suggestions that greatly helped to improve the writing of the paper. During this work, J.\ A.\ G.\ was supported by an SFI Omidyar Fellowship, K.\ D.\ M. by the NSF grant CCF-1017760, and Y.\ Q.\ by the Australian Research Council DECRA DE150100720.

\bibliographystyle{plain}
\bibliography{references}

\end{document}